%% file: main.tex
\documentclass[copyright,creativecommons]{eptcs}
\providecommand{\event}{Express/SOS 2025} 

\usepackage{iftex}

\ifpdf
  \usepackage[strings]{underscore}         
  \usepackage[T1]{fontenc}        
\else
  \usepackage{breakurl}           
\fi

\usepackage{graphicx} 
\usepackage{stmaryrd} 
\usepackage{relsize}
\usepackage{proof}
\usepackage{mathtools}
\usepackage[all,arc,cmtip]{xy}
\usepackage{amsmath,amssymb,amsthm}
\usepackage{mathrsfs}
\usepackage[mathcal]{eucal}
\input{macros}

\newcommand{\SPrg}{\mathrm{Pr}}
\newcommand{\Ch}{\mathrm{Ch}}
\usepackage[nointegrals]{wasysym}
\newtheorem{definition}{Definition}[section]  

\newtheorem{theorem}{Theorem}[section]
\newtheorem{lemma}{Lemma}[section]
\newtheorem{proposition}{Proposition}[section]
\newtheorem{corollary}{Corollary}[section]

\allowdisplaybreaks[2]

\title{An Adequacy Theorem Between Mixed Powerdomains and Probabilistic Concurrency
\\ (extended version)}
\author{Renato Neves
\institute{University of Minho \& INESC-TEC, Portugal}
\email{nevrenato@di.uminho.pt}}
\def\authorrunning{Renato Neves}
\def\titlerunning{An Adequacy Theorem Between Mixed Powerdomains and Probabilistic Concurrency}
\begin{document}
\maketitle

\begin{abstract}
        We present an adequacy theorem for a concurrent extension of
        probabilistic \textsc{GCL}.  The underlying denotational semantics is
        based on the so-called mixed powerdomains, which combine
        non-determinism with probabilistic behaviour. The theorem itself is formulated
        via M. Smyth’s idea of treating observable properties as open sets of a
        topological space. The proof hinges on a `topological generalisation'
        of K\"onig's lemma in the setting of probabilistic programming (a
        result that is proved in the paper as well).

        One application of the theorem is that it entails semi-decidability
        w.r.t.  whether a concurrent program satisfies an observable property
        (written in a certain form).  This is related to M. Escardó's
        conjecture about semi-decidability w.r.t. may and must probabilistic
        testing.
\end{abstract}

\section{Introduction}
\label{sec:intro}
\textbf{Motivation.} The combination of probabilistic operations with
concurrent ones  -- often referred to as probabilistic concurrency -- has been
extensively studied in the last decades and has various
applications~\cite{jones89,baier00,mislove00,varacca07,mciver13,visme19}. A
central concept in this paradigm is that of a (probabilistic)
scheduler~\cite{escardo04,varacca07}: \ie\ a memory-based entity that analyses
and (probabilistically) dictates how programs should be interleaved. 

Albeit intensively studied and with key results probabilistic concurrency still
carries fundamental challenges.  In this paper we tackle the following
instance: take a concurrent \emph{imperative} language equipped with probabilistic
operations. Then ask whether a given program $P$ in state $s$ satisfies a
property $\phi$, in symbols $\pv{P}{s} \models \phi$.  For example $\phi$ may
refer to the fact that $\pv{P}{s}$ terminates with probability greater than
$\frac{1}{2}$ under \emph{all} possible schedulers (which is intimately connected
to the topic of statistical termination~\cite{hart83,lengal17}). Or dually
$\phi$ could refer to the existence of \emph{at least one} scheduler under
which $\pv{P}{s}$ terminates with probability greater than $\frac{1}{2}$. Such
a question can be particularly challenging to answer, for it involves
quantifications over the universe of scheduling systems, which as we shall see
are \emph{uncountably} many.  An even more demanding type of question arises from
establishing an observational preorder $\lesssim$ between programs in state
$s$. More specifically we write,
\[
        \pv{P}{s} \lesssim \pv{Q}{s} \, \, \text{ iff } \, \, \Big (\text{for all properties $\phi$.}
        \, \, \pv{P}{s} \models \phi \text{ implies } \pv{Q}{s} \models \phi
        \Big )
\]
and ask whether $\pv{P}{s} \lesssim \pv{Q}{s}$ for two programs $P$ and $Q$ in
state $s$. Adding to the quantifications over scheduling systems, we now have a
universal quantification over formulae.

\smallskip
\noindent
\textbf{Contributions.}
In this paper we connect the previous questions to domain-theoretic
tools~\cite{gierz03,larrecq13}.  Concretely we start by recalling necessary
background (Section~\ref{sec:back}). Then in order to provide a formal
underpinning to these questions, we introduce a concurrent extension of the
probabilistic guarded command language (pGCL)~\cite{mciver05} and equip it with
an operational semantics as well as a logic for reasoning about program
properties (Section~\ref{sec:lang}). The logic itself is based on M.  Smyth's
idea of treating observable properties as open sets of a topological
space~\cite{smyth83,vickers89}.  We then introduce a domain-theoretic
denotational semantics $\sem{-}$ (Section~\ref{sec:den}) and establish its
computational adequacy w.r.t.  the operational counterpart
(Section~\ref{sec:adeq}). More formally we establish an equivalence of the
form,
\[
        \pv{P}{s} \models \phi \text{ iff }
        \sem{P}(s) \models \phi
        \hspace{3.5cm}
        \text{(for all properties $\phi$ in our logic)}
\]
whose full details will be given later on. From adequacy we then
straightforwardly derive a \emph{full abstraction} result w.r.t. the
observational preorder $\lesssim$, \ie\ we obtain the equivalence,
\[
        \pv{P}{s} \lesssim \pv{Q}{s} \text{ iff } \sem{P}(s) \leq
        \sem{Q}(s)
\]
for all programs $P$ and $Q$ and state $s$. 

An interesting and somewhat surprising consequence of our adequacy theorem is
that the question of whether $\pv{P}{s} \models \phi$ becomes semi-decidable
for a main fragment of our logic. This includes for example the two questions
above concerning the termination of $\pv{P}{s}$ with probability greater than
$\frac{1}{2}$. What is more, semi-decidability is realised by an exhaustive
search procedure -- which looks surprising because as already mentioned the
space of schedulers is uncountable and the question of whether $\pv{P}{s}
\models \phi$ involves quantifications over this space.  In
Section~\ref{sec:conc} we further detail this aspect and also the topic of
statistical termination~\cite{hart83,lengal17,majumdar24,fu24}.  Also in
Section~\ref{sec:conc}, to illustrate the broad range of applicability
of our results we show  how to
instantiate them to the quantum setting~\cite{feng11,ying18,ying22}. Briefly
the resulting concurrent language treats measurements and qubit resets as
probabilistic operations which thus frames the language in the context of
probabilistic concurrency. To the best of our knowledge the quantum concurrent
language obtained is the first one equipped with an adequacy theorem.  We
then conclude with a discussion of future work.  

\smallskip
\noindent
\textbf{Outline.}
We proceed by outlining our denotational semantics.  First the fact that we are
working with concurrency, more specifically with the interleaving paradigm,
suggests the adoption of the `resumptions model' for defining the semantics.
This approach was proposed by R. Milner in the seventies~\cite{milner75} and
subsequently adapted to different contexts (see
\eg~\cite{stark96,pirog14,goncharov16}). When adapted to our
specific case program denotations $\asem{P}$ are intensional and defined as
elements of the \emph{greatest} solution of the equation,
\begin{align}
        \label{eq:domeq}
        X \cong T(S + X \times S)^S
\end{align}
where $S$ is a pre-determined set of states and $T$ a functor that encodes a
combination of non-determinism (which handles interleaving in concurrency) and
stochasticity. In other words such program denotations are elements of (the
carrier of) a final coalgebra. 

Second we will see that in order to obtain adequacy w.r.t.  the operational
semantics  one needs to extensionally collapse $\asem{-}$ -- \ie\ given a
program $P$ one needs to find a suitable way of keeping only the terminating
states of $\asem{P}$ whilst accounting for the branching structure encoded in
$T$. Denoting $X$ as the greatest solution of Equation~\eqref{eq:domeq}, such
amounts to finding a suitable morphism $\mathrm{ext} : X \to T(S)^S$. Our
approach to this is inspired by~\cite{hennessy79}: such a morphism arises
naturally if $T$ is additionally a monad and $X$ is additionally the
\emph{least} solution of Equation~\eqref{eq:domeq}. Concretely we obtain the
following commutative diagram,
\begin{align}
        \label{eq:diag}
        \xymatrix@C=90pt{
                \Pr 
                \ar@{.>}[r]^(0.45){\asem{-}} 
                \ar[d] \ar@{}[d]_{ \mathrm{op} }
                &
                X 
                \ar@{.>}[r]^(0.45){\mathrm{ext}}
                \ar@/^10pt/[d]^{\unfold}
                &
                T(S)^S
                \\
                T(S + \Pr \times S)^S 
                \ar[r]_(0.5){T(\id + \asem{-} \times \id)^S} 
                & 
                T(S + X \times S)^S  
                \ar[r]_(0.45){T(\id + \mathrm{ext} \times \id)^S} 
                \ar@/^10pt/[u]^{\fold}
                \ar@{}[u]|-{\cong}
                & 
                T(S + T(S)^S \times S)^S
                \ar[u]_{{[\eta^T,\ \app]^\star}^S}
        }
\end{align}
where $\Pr$ denotes the set of programs, the operations $\eta^T$, $(-)^\star$
refer to the monadic structure postulated for $T$, $\app$ denotes the
application map, and $\mathrm{op}$ denotes our operational semantics in 
coalgebraic form. The fact that the left rectangle commutes (proved in
Theorem~\ref{theo:sem_eq}) means that the intensional semantics $\asem{-}$
agrees with the `one-step' transitions that arise from the operational
semantics. The fact that the right rectangle commutes is obtained by
construction (initiality) and yields a recursive definition of the extensional
collapse, with elements of $X$ seen as `resumption trees'.  Our final
denotational semantics $\sem{-}$ is the composition $\mathrm{ext} \comp
\asem{-}$.

In what regards the choice of $T$ there are three natural candidates which are
collectively known as mixed
powerdomains~\cite{mciver01,mislove04,tix09,keimel09}. One is associated with
angelic (or may) non-determinism, another with demonic (or must)
non-determinism, and the other one with both cases. Our computational adequacy
result applies to these three powerdomains, a prominent feature of our result
being that the aforementioned logic (for reasoning about properties) varies
according to whatever mixed powerdomain is chosen. Notably keeping in touch
with its name, the proof for demonic non-determinism is much harder to achieve
than the angelic case: in order to prove the former we establish what can
be intuitively seen as a topological generalisation of K\"onig's
lemma~\cite{franchella97} in the setting of probabilistic programming. This is
detailed further in the paper and proved in appendix (Section~\ref{sec:konig})
after a series of auxiliary results.

\smallskip
\noindent
\textbf{Related work.}
The sequential fragment of our language is the well-known pGCL~\cite{mciver05},
a programming language that harbours both probabilistic and non-deterministic
choice.  The idea of establishing adequacy between pGCL and the three mixed
powerdomains was propounded in~\cite{keimel11} as an open endeavour -- and in
this regard the adequacy of the angelic mixed powerdomain was already
established in~\cite{varacca07}.  More recently  J.  Goubault-Larrecq
\cite{larrecq19} established computational adequacy at the unit type
between an extension of PCF with probabilistic and non-deterministic choice and
the three mixed powerdomains.  Note however that in this case the associated
operational semantics is not based on probabilistic schedulers like
in~\cite{varacca07} and our case, but rather on what could be called `angelic
and demonic strategies'. Moreover to scale up the adequacy result in the
\emph{op. cit.} to a global store (as required in pGCL) demands that states can
be directly compared in the language, an assumption that we do not make and
that does not hold for example in the quantum setting. We therefore believe
that our adequacy result restricted to the sequential fragment is already
interesting \emph{per se}.

Adequacy at the level of probabilistic concurrency has been studied at various
fronts, however as far as we are aware usually in the setting of process
algebra~\cite{baier00,mislove04,visme19}. This is a direction orthogonal to
ours, in the sense that the latter focusses on interactions and is 
intensional in nature; whilst here we work in the setting of imperative
programming and thus take a more extensional approach.

Semi-decidability concerning may and must statistical termination, in a spirit
similar to ours, was conjectured in~\cite{escardo09}. Note however that the
host programming language in the latter case was a \emph{sequential,
higher-order} language with non-deterministic and probabilistic choice. The
underlying scheduling system in the \emph{op. cit.} was also different from
ours: it was encoded via pairs of elements in the Cantor space $2^\Nats$.

\smallskip
\noindent
\textbf{Prerequisites and notation.}
We assume from the reader basic knowledge of category theory, domain theory,
and topology (details about these topics are available for example
in~\cite{adamek09,gierz03,larrecq13}). We denote the left and right injections
into a coproduct by $\inl$ and $\inr$, respectively, and the application map
$X^Y \times Y \to X$ by $\app$. Given a set $X$ we use $X^\ast$ to denote the
respective set of lists, and given an element $x \in X$ we denote the `constant
on $x$' map $1 \to X$  by $\const{x}$. We omit subscripts in natural
transformations whenever no ambiguity arises from this. We use $\eta$ and
$(-)^\star$ to denote respectively the unit and lifting of a monad
$(T,\eta,(-)^\star)$, and whenever relevant we attach $T$ as a superscript to
these two operations to clarify which monad we are referring to. We denote the
set of positive natural numbers $\{1,2,\dots\}$ by $\Nats_+$.  Omitted proofs
are found in the paper's appendix.

\section{Background}
\label{sec:back}

\smallskip
\noindent
\textbf{Domains and topology.}
We start by recalling pre-requisite notions for establishing the aforementioned
denotational semantics and corresponding adequacy theorem.  As usual we call
DCPOs those partially ordered sets that are directed-complete. We call domains
those DCPOs that are continuous~\cite[Definition I-1.6]{gierz03}. We call a
DCPO pointed it it has a bottom element and we call a map between pointed DCPOs
strict if it preserves bottom elements. Similarly we call a functor on pointed
DCPOs strict if it preserves strict maps. Given a poset $X$ its Scott topology
$\sigma(X)$ consists of all those subsets $U \subseteq X$ that are upper-closed
and inaccessible by directed joins: \ie\ for every directed join $\bigvee_{i
\in I} x_i \in U$ there must already exist some element $x_i$ in the family
$(x_i)_{i \in I}$ such that $x_i \in U$.  Another important topology on $X$ is
the lower topology $\omega(X)$. It is generated by the subbasis of closed sets
$\{\upclos x \mid x \in X \}$. Yet another important topology is the Lawson
topology $\lambda(X)$ which is generated by the subbasis $\sigma(X) \cup
\omega(X)$ (see details in~\cite[Section III-1]{gierz03}). When treating a
poset $X$ as a topological space we will be tacitly referring to its Scott
topology unless stated otherwise. A domain $X$ is called coherent if the
intersection of two compact saturated subsets in $X$ is again compact. Finally
note that every set can be regarded as a coherent domain by taking the discrete
order. We will often tacitly rely on this fact.

Let $\Dcpo$ be the category of DCPOs and continuous maps. Let $\Dom$ and $\Coh$
be the full subcategories of $\Dcpo$ whose objects are respectively domains and
coherent domains. We thus obtain the chain of inclusions $\Coh \hookrightarrow
\Dom \hookrightarrow \Dcpo$. Next, let $\catC$ be any full subcategory of
$\Dcpo$.  Since it is a full subcategory the inclusion $\catC \hookrightarrow
\Dcpo$ reflects limits and colimits~\cite[Definition 13.22]{adamek09}. This
property is very useful in domain theory because domain theoreticians often
work in full subcategories of $\Dcpo$. It is well-known for example that
$\Dcpo$-products of (coherent) domains where only finitely many are non-pointed
are (coherent) domains as well~\cite[Proposition 5.1.54 and Exercise
8.3.33]{larrecq13}.  It follows that such limits are also limits in $\Dom$, and
if the involved domains are coherent then they are limits in $\Coh$ as well.
The same reasoning applies to coproducts: $\Dcpo$-coproducts of (coherent)
domains are (coherent) domains as well~\cite[Proposition 5.1.59 and Proposition
5.2.34]{larrecq13}. Finally observe that $\Dcpo$ is distributive and that this
applies to any full subcategory of $\Dcpo$ that is closed under binary
(co)products. This includes for example $\Dom$ and $\Coh$.

\smallskip
\noindent
\textbf{The probabilistic powerdomain.}
Let us now detail the  probabilistic powerdomain and associated constructions.
These take a key rôle not only in the definition of the three mixed
powerdomains, but also in our denotational semantics and corresponding adequacy
theorem. We start with the pre-requisite notion of a continuous valuation.

\begin{definition}[\cite{tix09,goubault20}]
        \label{defn:pp}
Consider a topological space $X$ and let $\mathcal{O}(X)$ be its topology.  A
function $\mu : \mathcal{O}(X) \to [0,\infty]$ is called a continuous valuation
on $X$ if for all opens $U,V \in \mathcal{O}(X)$ it satisfies the following
conditions:
\begin{itemize} 
        \item $\mu(\emptyset) = 0$;
        \item $U \subseteq V \Rightarrow \mu(U) \leq \mu(V)$;
        \item $\mu(U) + \mu(V) = \mu(U \cup V) + \mu(U \cap V)$;
        \item $\mu \left (\bigcup_{i \in I} U_i \right ) = \bigvee_{i \in I}
                \mu(U_i)$ for every directed family of opens $(U_i)_{i \in I}$.
\end{itemize}
\end{definition}

\smallskip
\noindent
We use $\PP(X)$ to denote the set of continuous valuations on $X$. We also use
$\PP_{=1}(X)$ and $\PP_{\leq 1}(X)$ to denote respectively the subsets of
continuous valuations $\mu$ such that $\mu(X) = 1$ and $\mu(X) \leq 1$.  An
important type of continuous valuation is the \emph{point-mass} (or Dirac)
valuation $\delta_x$ $(x \in X)$ defined by,
\[
        \delta_x(U) = \begin{cases}
                1 & 
                \text{ if }  x \in U \\
                0 & 
                \text{ otherwise }
        \end{cases}
\]
Recall that a cone is a set $C$ with operations for addition $+ : C \times C
\to C$ and scaling $\cdot : \Rz \times C \to C$ that satisfy the laws of vector
spaces except for the one that concerns additive inverses~\cite[Section
2.1]{tix09}.  It is well-known that $\PP(X)$ forms a cone with scaling and
addition defined pointwise. If $X$ is a domain then $\PP(X)$ is also a domain.
The order on $\PP(X)$ is defined via pointwise extension and a basis (in the
domain-theoretic sense) is given by the finite linear combinations $\sum_{i \in
I} r_i \cdot \delta_{x_i}$ of point-masses with $r_i \in \Rz$ and $x_i \in X$.
We will often abbreviate such combinations to $\sum_i r_i \cdot x_i$. Moreover
we will use $\PP_{=1, \omega}(X)$ and $\PP_{\leq 1, \omega}(X)$ to denote
respectively the subsets of continuous valuations in $\PP_{=1}(X)$ and
$\PP_{\leq 1}(X)$ that are finite linear combinations of point-masses.
Whenever $X$ is a domain the Scott-topology of the domain $\PP(X)$ coincides
with the so-called \emph{weak topology}~\cite{goubault20} which is generated by
the subbasic sets,
\[
        \statt U = \{ \mu \in \PP(X) \mid \mu(U) > p \}
        \hspace{2.5cm} (U \text{ an open of } X \text{ and } p \in \Rz)
\]
The probabilistic powerdomain also gives rise to a category of cones.  We
briefly detail it next, and direct the reader to the more thorough account
in~\cite[Chapter 2]{tix09}.

\begin{definition}
        A d-cone $(C,\leq,\,\cdot\,,+)$ is a cone $(C,\,\cdot\,,+)$ such
        that the pair $(C,\leq)$ is a DCPO and the operations $\cdot\, : \Rz
        \times C \to C$ and $+ : C \times C \to C$ are Scott-continuous.  If
        the ordered set $(C,\leq)$ is additionally a domain then we speak of a
        continuous d-cone. The category $\Cone$ has as objects continuous
        d-cones and as morphisms continuous linear maps.
\end{definition}
The forgetful functor $\Forg{} : \Cone \to \Dom$ is right adjoint: the
respective universal property is witnessed by the construct $\PP(-)$. 
Specifically for every domain $X$, continuous d-cone $C$, and $\Dom$-morphism
$f : X \to \Forg{}(C)$ we have the diagrammatic situation,
\[
        \xymatrix@C=35pt{
                X \ar[r]^(0.45){x \mapsto \delta_x} \ar[dr]_(0.45){f} & \Forg{} \PP(X)
                \ar[d]^{\Forg{} (f^\star)}
                \\
                                         & \Forg{} (C)
        }
        \hspace{2cm}
        \xymatrix{
                \PP(X) \ar@{.>}[d]^{f^\star} \\
                C
        }
\]
with $f^\star$ defined by $f^\star(\sum_i r_i \cdot x_i) = \sum_i r_i \cdot
f(x_i)$ on basic elements $\sum_i r_i \cdot x_i \in \PP(X)$.  We thus obtain
standardly a functor $\PP : \Dom \to \Cone$. Now, the functor $\Forg{} : \Cone
\to \Dom$ is additionally monadic, with the respective left adjoint given by
$\PP : \Dom \to \Cone$. Among other things such implies that $\Cone$ is as
complete as the category $\Dom$. For example $\Cone$ has all products of
continuous d-cones (recall our previous remarks about $\Dom$ and note that
d-cones are always pointed~\cite[page 21]{tix09}).  Note as well that binary
products are actually biproducts by virtue of addition being Scott-continuous.
Finally let $\LCone$ be the full subcategory of $\Cone$ whose objects are
Lawson-compact (\ie\ compact in the Lawson topology). $\PP(X)$ is
Lawson-compact  whenever $X$ is a coherent domain~\cite[Theorem 2.10]{tix09}
and Lawson-compactness of a domain entails coherence~\cite[Theorem
III-5.8]{gierz03}. We thus obtain a functor $\PP : \Coh \to \LCone$ that gives
rise to the commutative diagram,
\[
\xymatrix@C=35pt{
                \ar@<1mm>@/^3mm/[d]^{\Forg{}}\LCone  \,
                \ar@{}[d]|{\dashv}
                \ar@{>->}[r]
                & 
                \Cone
                \ar@<1mm>@/^3mm/[d]^{\Forg{}} 
                \\
                \Coh \, \ar@<1mm>@/^3mm/[u]^{\PP}
                \ar@{>->}[r]
                & 
                \Dom \ar@<1mm>@/^3mm/[u]^{\PP}
                \ar@{}[u]|{\dashv}
        }
\]
The following result, which we will use multiple times, is somewhat folklore.
Unfortunately we could not find a proof in the literature, so we provide
one here.
\begin{theorem}
        \label{theo:locin}
        Consider a zero-dimensional topological space $X$ (\ie\ a topological
        space with a basis of clopens). For all valuations $\mu,\nu \in
        \PP_{=1}(X)$ if $\mu \leq \nu$ then $\mu = \nu$.  In other words the
        order on $\PP_{=1}(X)$ is discrete.
\end{theorem}

\begin{proof}
        Observe that if $\mathcal{B}$ is a basis of clopens of $X$ then the set
        $\mathcal{B}^\vee = \{ U_1 \cup \dots \cup U_n \mid U_1, \dots, U_n \in
        \mathcal{B} \}$ is directed and it is constituted only by clopens as well.
        Next we reason by contradiction and by appealing to the conditions
        imposed on valuations (Definition~\ref{defn:pp}). Suppose that there
        exists an open $U \subseteq X$ such that $\mu(U) < \nu(U)$. This open
        can be rewritten as a directed union $\bigcup\,  \mathcal{D}$ where
        $\mathcal{D} = \{ V \subseteq U \mid V \in \mathcal{B}^\vee \}$ and
        we obtain the strict inequation,
        \[
                \mu(U) < \bigvee_{V \in \mathcal{D}} \nu(V)
        \]
        The latter entails the existence of a clopen $V \in \mathcal{D}$ such
        that $\mu(V) < \nu(V)$. It then must be the case that $1 = \mu(X)
        = \mu(V) + \mu(X\backslash V) < \nu(V) + \nu(X\backslash V) = \nu(X)$.
        This proves that $1 < \nu(X)$, a contradiction.
\end{proof}

\smallskip
\noindent
\textbf{The three mixed powerdomains.}
We now present the so-called geometrically convex powercones $\Pow_x(-)$ ($x
\in \{l,u,b\})$, \emph{viz.} convex lower $(\Pow_l)$, convex upper $(\Pow_u)$,
and biconvex $(\Pow_b)$~\cite[Chapter 4]{tix09}. As we will see later on, their
composition with the probabilistic powerdomain (which was previously recalled)
yields the three mixed powerdomains that were mentioned in the introduction.

We start with some preliminary notions.  Given a DCPO $X$ and a subset $A
\subseteq X$ we denote the Scott-closure of $A$ by $\overline{A}$.  The subset
$A$ is called \emph{order convex} whenever for all $x \in X$ if there exist
elements $a_1,a_2 \in A$ such that $a_1 \leq x \leq a_2$ then $x \in A$. Given
a cone $C$ we call a subset $A \subseteq C$ geometrically convex
(or just convex) if $a_1,a_2 \in A$ entails $p \cdot a_1 + (1-p) \cdot a_2 \in
A$ for all $p \in [0,1]$.  We denote the convex closure of $A$ by $\conv{A}$.
The latter is explicitly defined by,
\[
        \conv{A} = \left \{ \textstyle{ \sum_{i \in I}} \, p_i \cdot a_i \mid
                \textstyle{ \sum_{i \in I}} \, p_i = 1 \> 
                \text{and} \> \forall i \in I. \, a_i \in A 
        \right \}
\]
For two finite subsets $F,G \subseteq C$ the expression $F+G$ denotes
Minkowski's sum and $p \cdot F$ denotes the set $\{p \cdot c \mid c \in F \}$.
If a subset $A \subseteq C$ is non-empty, Lawson-compact, and order convex we
call it a lens.  Given a continuous d-cone $C$ we define the following
partially ordered sets,
\begin{align*}
        \Pow_l(C) & = \{ A \subseteq C \mid A \text{ non-empty, closed, and convex} \}
                          \\
        \Pow_u(C) & = \{ A \subseteq C \mid A \text{ non-empty, compact, saturated,
        and convex} \}
                          \\
        \Pow_b(C) & = \{ A \subseteq C \mid A \text{ a convex lens}
        \}
\end{align*}
with the respective orders defined by $A \leq_l B \text{ iff }
\mathord{\downarrow} A \subseteq \mathord{\downarrow} B$, $A \leq_u B \text{
iff } \mathord{\uparrow} B \subseteq \mathord{\uparrow} A$, and $A \leq_b B
\text{ iff } \mathord{\downarrow} A \subseteq \mathord{\downarrow} B \text{ and
} \mathord{\uparrow}B \subseteq \mathord{\uparrow} A$. Whenever working with
$\Pow_b(C)$ we will assume that $C$ is additionally Lawson-compact.  All three
posets form domains and $\Pow_b(C)$ is additionally Lawson-compact. In
particular the sets of the form $\overline{\conv{F}}$, $\mathord{\uparrow}
\conv{F}$, and $\overline{\conv{F}} \cap \, \mathord{\uparrow}  \conv{F}$ (for
$F$ a finite set) form respectively a basis (in the domain-theoretic sense) for
the convex lower, convex upper, and biconvex powercones.  For simplicity we
will often denote $\overline{\conv{F}}$ by $l(F)$, $\mathord{\uparrow}
\conv{F}$ by $u(F)$, and $\overline{\conv{F}} \cap \mathord{\uparrow} \conv{F}$
by $b(F)$ -- and in the general case \ie\ when we wish to speak of all three
cases at the same time we write $x(F)$.

The Scott topologies of the three powercones have well-known explicit
characterisations that are closely connected to the `possibility' ($\angel$)
and `necessity' ($\demon$) operators of modal logic~\cite{winskel85,mislove04}.
Specifically the Scott topology of $\Pow_b(C)$ is generated by the subbasic
sets,
\[
        \angel  U  = \{ A \mid A \cap U \not = \emptyset \}
        \> \text{ and } \>
        {\demon} \, U   = \{ A \mid A \subseteq U\}
        \hspace{2cm}
        (U \text{ an open of } C)
\]
the Scott topology of $\Pow_l(C)$ is generated only by those subsets of the
form $\angel U$ and the Scott topology of $\Pow_u(C)$ is generated only by
those subsets of the form ${\demon}\, U$. When seeing open subsets as
observable properties~\cite{smyth83,vickers89}, the powercone $\Pow_l(C)$ is
thus associated with the \emph{possibility} of a given property $U$ being true
(\ie\ whether a property $U$ holds in at least one scenario) $\angel U$, the
powercone $\Pow_u(C)$ is associated with the \emph{necessity} of a given
property $U$ being true (\ie\ whether a property $U$ holds in all scenarios)
$\demon U$ and the powercone $\Pow_b(C)$ with both cases.  These observations
together with the previous characterisation of the Scott topology of the
probabilistic powerdomain yield a natural logic for reasoning about program
properties in concurrent pGCL, as detailed in the following section.

All three powercones are equipped with the structure of a cone and with a
continuous semi-lattice operation $\uplus$: the whole structure is defined on
basic elements by,
\[
x(F) + x(G)  = x(F + G) 
\hspace{1.5cm}
r \cdot x(F) = x (r \cdot F)
\hspace{1.5cm}
x(F) \uplus x(G) = x(F \cup G)
\]
and is thoroughly detailed in~\cite[Section 4]{tix09}. Moreover the convex
lower and convex upper variants form monads in $\Cone$ whereas the biconvex
variant forms a monad in $\LCone$. In all three cases the unit $C \to
\Pow_x(C)$ is defined by $c \mapsto x(\{c\})$ and the Kleisli lifting of  $f :
C \to \Pow_x(D)$ is defined on basic elements by $f^\star(x(F)) =
\mathlarger{\uplus} \left \{ f(a) \mid a \in F  \right \}$. Finally we often
abbreviate the composition $\Pow_x \PP$ simply to $\Pow \PP$ and by a slight
abuse of notation treat $\Pow_x \PP$ as a functor $\Pow_x \PP : \Dom \to \Dom$
if $x \in \{l,u\}$ or as a functor $\Pow_x \PP : \Coh \to \Coh$ if $x \in
\{b\}$. These three functors constitute the three mixed powerdomains.  It
follows from composition of adjunctions that every mixed powerdomain is a monad
on $\Dom$ (resp.  $\Coh$). 

\begin{theorem}
        \label{theo:str} For every $x \in \{l,u\}$ the functor is $\Pow_x \PP :
        \Dom \to \Dom$ is strong.  Given two domains $X$ and $Y$
        the tensorial strength $\mathrm{str} : X \times \Pow_x \PP (Y)
        \to \Pow_x \PP(X \times Y)$ is defined on basic elements as the mapping
        $(a,x(F)) \mapsto x(\{ \delta_a \otimes \mu \mid \mu \in F\})$ where,
        \[
                \delta_a \otimes \textstyle{\sum_i} \, p_i \cdot y_i =
                \textstyle{\sum_i} \, p_i \cdot (a,y_i)
        \]
        for every linear combination $\sum_i \, p_i \cdot y_i$.  The same
        applies to the functor $\Pow_b \PP$.
\end{theorem}
Similarly the following result is also obtained straightforwardly.
\begin{proposition}
        The monads $\Pow_x : \Cone \to \Cone$ ($x \in \{l,u\})$ are strong
        w.r.t. coproducts. The corresponding natural
        transformation $\mathrm{str} : \Id + \Pow_x \to \Pow_x(\Id + \Id)$ is
        defined as $\mathrm{str} = [\eta \comp \inl , \Pow_x \inr]$. An
        analogous result applies to the biconvex powercone $\Pow_b : \LCone
        \to \LCone$.
\end{proposition}

By unravelling the definition of $\mathrm{str}$ and by taking advantage of the
fact that binary coproducts of cones are also products (recall our previous
observations about biproducts in $\Cone$) we obtain the following equivalent
formulation of $\mathrm{str}$ for all components $C,D$,
\[
                \mathrm{str}_{C,D}: C \times \Pow_x(D) \to \Pow_x(C \times D)
                \qquad \qquad
                (c,x(F)) \mapsto x(\{c\} \times F)
\]

\section{Concurrent pGCL, its operational semantics, and its  logic}
\label{sec:lang}

As mentioned in Section~\ref{sec:intro} our language is a concurrent extension
of pGCL~\cite{mciver01,mciver05}. It is described by the BNF grammar,
\begin{align*}
        \label{grammar_lang}
        P ::= \prog{skip} 
        \mid 
        \prog{a} 
        \mid 
        P \, ; \, P
        \mid
        P \parallel P
        \mid 
        P \, +_\prog{p} \, P
        \mid 
        P \, + \, P
        \mid
        \prog{if} \> \prog{b} \> \prog{then} \> \> P \> \> \prog{else} \> \> P
        \mid
        \prog{while} \> \prog{b} \> \> P
\end{align*}
where $\prog{a}$ is a program from a pre-determined set of programs and
$\prog{b}$ is a condition from a pre-determined set of conditions.  A program
$P +_\prog{p} Q$ $(\prog{p} \in [0,1] \cap \mathbb{Q})$ represents a
probabilistic choice between either the evaluation of $P$ or $Q$. $P + Q$ is
the non-deterministic analogue.  The other program constructs are quite
standard so we omit their explanation.

We now equip the language with a small-step operational semantics. The latter
will be used later on for introducing a big-step counterpart via probabilistic
scheduling.  First we take an arbitrary set $S$ of states, for each atomic
program $\prog{a}$ we postulate the existence of a function $\sem{\prog{a}} : S
\to \Dist(S)$, and for each condition $\prog{b}$ we postulate the existence of
a function $\sem{\prog{b}} : S \to \{ \mathtt{tt,ff} \}$. Let $\SPrg$ be the
set of programs. The language's small-step operational semantics is the
relation $\longrightarrow \, \subseteq (\SPrg \times S) \times \Dist (S +
(\SPrg \times S) )$ that is generated inductively by the rules in
Figure~\ref{fig:op_sem}. The latter are a straightforward probabilistic
extension of the rules presented in the classical reference \cite[Chapters 6
and 8]{reynolds98}, so we omit their explanation.

Let us denote by $\pv{P}{s} \longrightarrow$ the set $\{ \mu \mid \pv{P}{s}
\longrightarrow \mu \}$.  The small-step operational semantics enjoys the
following key property for adequacy.

\begin{proposition}
        \label{theo:finitely}
        For every program $P$ and state $s$ the set
        $\pv{P}{s} \longrightarrow$ is finite.
\end{proposition}
\begin{proof}
        Follows by induction over the syntactic structure of programs
        and by taking into account that finite sets are closed under
        binary unions and functional images.
\end{proof}

\begin{figure*}
        \centering
\scalebox{0.98}{
\renewcommand{\arraystretch}{3}
\begin{tabular}{p{0cm}  p{0cm}  p{0cm}}
        \multicolumn{1}{c}{
        \infer{ \pv{\prog{a}}{s} \longrightarrow \sem{\prog{a}}(s) }{}
        \hspace{1cm} 
        }
        & 
        \multicolumn{1}{c}{
        \infer{ \pv{\prog{skip}}{s} \longrightarrow 1 \cdot s }{}
        \hspace{1cm} 
        }
        &
        \multicolumn{1}{c}{
        \infer{ \pv{P; Q}{s} 
        \longrightarrow \sum_i p_i \cdot \pv{P_i; Q}{s_i} + \sum_j p_j \cdot \pv{Q}{s_j}}
                { \pv{P}{s} 
                \longrightarrow \sum_i p_i \cdot \pv{P_i}{s_i} + \sum_j p_j \cdot s_j}
        }
        \\[10pt]
        \multicolumn{3}{c}{
        \infer{ \pv{P \parallel Q}{s} 
        \longrightarrow \sum_i p_i \cdot \pv{P_i \parallel Q}{s_i} 
                + \sum_j p_j \cdot \pv{Q}{s_j}}
                { \pv{P}{s} 
                \longrightarrow \sum_i p_i \cdot \pv{P_i}{s_i} + \sum_j p_j \cdot s_j}
        }
        \\[10pt]
        \multicolumn{3}{c}{ 
        \infer{ \pv{P \parallel Q}{s} 
        \longrightarrow \sum_i p_i \cdot \pv{P \parallel
        Q_i}{s_i} + \sum_j p_j \cdot \pv{P}{s_j}}
                { \pv{Q}{s} 
                \longrightarrow \sum_i p_i \cdot \pv{Q_i}{s_i} + \sum_j p_j \cdot s_j }
        }
        \\[10pt]
        \multicolumn{2}{l}
        {
                \infer{ \pv{P \> +_\prog{p} \> Q}{s} \longrightarrow 
                \prog{p} \cdot \mu + (1-\prog{p}) \cdot \nu }{
                  \pv{P}{s} \longrightarrow \mu \qquad
                \pv{Q}{s} \longrightarrow \nu}
        }
        &
        \multicolumn{1}{l}
        {
        \infer{ 
                \pv{P + Q}{s} \longrightarrow \mu 
        }
        {
                \pv{P}{s} \longrightarrow \mu
        }
        \hspace{1.8cm}
        \infer{ 
                \pv{P + Q}{s} \longrightarrow \mu 
        }
        {
                \pv{Q}{s} \longrightarrow \mu         
        }
        }
        \\[10pt]
        \multicolumn{2}{c}
        {
                \infer{
                \pv{\prog{if \, \> b} \> 
                \, \prog{then} \, \> P \> \, \prog{else} \, \> Q}{s}
                \longrightarrow 1 \cdot \pv{P}{s} 
                }
                {
                \sem{\prog{b}}(s) = \mathtt{tt} 
                }
        }
        &
        \multicolumn{1}{c}{
                \infer{ 
                \pv{\prog{if \, \> b} \> \, \prog{then} \, 
                \> P \> \, \prog{else} \,  \> Q}{s} 
                \longrightarrow 1 \cdot \pv{Q}{s}  
                }
                {
                \sem{\prog{b}}(s) = \mathtt{ff}
                }
        }
        \\[10pt]
        \multicolumn{2}{c}
        {
                \infer{ 
                \pv{\prog{while \, \> b} \> \, P}{s} \longrightarrow 
                1 \cdot \pv{P; \prog{while \, \> b} \> \, P}{s} 
                }
                {
                \prog{\sem{b}}(s) = \mathtt{tt} 
                }
        }
        &
        \multicolumn{1}{c}
        {
                \infer{ 
                \pv{\prog{while \, \> b} \, \> P}{s} \longrightarrow 1 \cdot s 
                }
                {
                \sem{\prog{b}}(s) = \mathtt{ff}
                }
        }
\end{tabular}
}

\caption{Small-step operational semantics}
\label{fig:op_sem}
\end{figure*}

We now introduce a big-step operational semantics. As mentioned before it is
based on the previous small-step semantics and the notion of a probabilistic
scheduler~\cite{varacca07}.  Intuitively a scheduler resolves all
non-deterministic choices that are encountered along the evaluation of a
program $P$ based on a history $h$ of previous decisions and the current state
$s$.
\begin{definition}
        A probabilistic scheduler $\sch$ is a partial function,
        \[
                \sch : \big ((\mathrm{Pr} \times S) \times
                \Dist(S + (\mathrm{Pr} \times S))\big )^\ast \times (\mathrm{Pr} \times S)
                \xrightharpoonup{\hspace{0.3cm} }
                \Dist \Dist(S + (\mathrm{Pr} \times S))
        \]
        such that whenever $\sch(h,\pv{P}{s})$ is well-defined it is a distribution
        of valuations in $\pv{P}{s} \longrightarrow$ (\ie\ it is a
        distribution of the possible valuations of one-step transitions
        that originate from $\pv{P}{s}$). 
\end{definition}
Note that if $\sch(h,\pv{P}{s})$ is well-defined it must be a linear
combination $\sum_k p_k \, \cdot \nu_k$ of valuations $\nu_k$ of the form
$\sum_{i} p_{k,i} \cdot \pv{P_{k,i}}{s_{k,i}} + \sum_{j} p_{k,j} \cdot
s_{k,j}$. Every representation $\nu_k$ is essentially unique by virtue
of $S + (\mathrm{Pr} \times S)$ being discrete. Not only this, by
Theorem~\ref{theo:locin}  the space  $\Dist(S + (\mathrm{Pr}
\times S))$ will also be discrete and thus the representation $\sum_k p_k \cdot
\nu_k$ itself will be essentially unique as well. This is important for the
definition of the big-step operational semantics and Proposition~\ref{theo:det}
below.

We will often denote a pair of the form $(h,\pv{P}{s})$ simply by $h\pv{P}{s}$.
The big-step operational semantics is defined as a relation
$h\pv{P}{s}\Downarrow^{\mathcal{S},n} \mu$ where $n \in \Nats$ is a natural
number and $\mu \in \PDist(S)$ is a valuation.  This relation represents an
\emph{$n$-step partial} evaluation of $\pv{P}{s}$ w.r.t.  $\mathcal{S}$ and
$h$, and it is defined inductively by the rules in Figure~\ref{fig:bop_sem}.

\begin{proposition}
        \label{theo:det}
        Take a natural number $n \in \Nats$, configuration $\pv{P}{s}$, history
        $h$, and scheduler $\sch$.  The big-step operational semantics has
        the following properties:
        \begin{itemize}
                \item (\emph{Determinism})
                        if $h\pv{P}{s} \Downarrow^{\mathcal{S},n} \mu
                        \text{ and }
                        h\pv{P}{s} \Downarrow^{\mathcal{S},n} \nu$
                        then
                        $\mu = \nu$;
                \item (\emph{Monotonicity})
                        if $h\pv{P}{s} \Downarrow^{\mathcal{S},n} \mu
                        \text{ and } 
                        h\pv{P}{s} \Downarrow^{\mathcal{S},n+1} \nu
                        $
                        then
                        $\mu \leq \nu$.
        \end{itemize}
\end{proposition}
\begin{proof}
        Follows straightforwardly by induction over the natural numbers and by
        appealing to the fact that scaling and addition (of valuations)
        are monotone.
\end{proof}

\begin{figure*}
       \begin{gather*}
       \scalebox{0.98}{
       \infer{h{\pv{P}{s}} \Downarrow^{\mathcal{S},0} \bot}{}
       }
       \hspace{2cm}
       \scalebox{1}{
       \infer{
               h\pv{P}{s} \Downarrow^{\mathcal{S},n+1} 
               \sum_k p_k \cdot \left ( \sum_{i} p_{k,i} 
                       \cdot \mu_{k,i} + \sum_{j} p_{k,j} \cdot s_{k,j}
               \right )
       }
       {
               \sch(h\pv{P}{s})
               =
               \sum_k p_k  \cdot  \nu_k
               \qquad
               \forall k,i. \,
               h\pv{P}{s}\nu_k\pv{P_{k,i}}{s_{k,i}} \Downarrow^{\mathcal{S},n} \mu_{k,i}
       }
       }
       \end{gather*}
\caption{Big-step operational semantics}
\label{fig:bop_sem}
\end{figure*}

\smallskip
\noindent
\textbf{The logic.}
We now introduce the aforementioned logic for reasoning about program
properties. It is an instance of geometric propositional logic~\cite{vickers89}
that allows to express both must ($\demon$) and may ($\angel$) statistical
termination ($\statt$).  It is two-layered specifically its formulae $\phi$ are
given by,
\[
        \phi ::=  \demon \varphi \mid \angel \varphi \mid \phi \wedge \phi \mid
        \bigvee \phi \mid \bot \mid \top
        \hspace{2cm}
        \varphi ::= \> \statt U \mid
                    \varphi \wedge \varphi \mid
                    \bigvee \varphi \mid
                    \bot \mid \top
\] 
where $U$ is a subset of the discrete state space $S$ (\ie\ $U$ is an
`observable property' per our previous remarks) and $p \in \Rz$. Whilst the top
layer handles the non-deterministic dimension (\eg\ all possible interleavings)
the bottom layer handles the probabilistic counterpart. Note that the
disjunction clauses are not limited to a finite arity but are set-indexed -- a
core feature of geometric logic.  The expression $\statt U$ refers to a program
`terminating in $U$' with probability \emph{strictly} greater than $p$.  The
formulae $\demon \varphi$ and $\angel \varphi$ correspond to universal and
existential quantification respectively.  For example $\demon \statt U$ reads
as ``it is \emph{necessarily} the case that the program at hand terminates in
$U$ with probability strictly greater than $p$'' whilst $\angel \statt U$ reads
as ``it is \emph{possible} that the program at hand terminates in $U$ with
probability strictly greater than $p$''.  It may appear that the choice of such
a logic for our language is somewhat \emph{ad-hoc}, but one can easily see that
it emerges as a natural candidate after inspecting the Scott topologies of the
mixed powerdomains and recalling the motto `observable properties as open sets'
(recall the previous section). In this context conjunctions and disjunctions
are interpreted respectively as intersections and unions of open sets.  

Observe as well that this logic is conceptually different from usual logics in
probabilistic process algebra (see for example~\cite{deng15}). Indeed the
latter are more focussed on reasoning about probabilities of (labelled)
transitions, whilst in our case the idea is to reason about probabilities of
halting states. Remarkably, a logic similar in spirit to those in process
algebra could also be topologically generated, not by taking the underlying
topology of the mixed powerdomains (as we do), but the topology of the
`resumptions model' mentioned in the introduction. See more details about this
particular aspect in~\cite{mislove04}.

The next step is to present a satisfaction relation between pairs $\pv{P}{s}$
and formulae $\phi$. To this effect we recur to the notion of a non-blocking
scheduler which is presented next.
\begin{definition}
        Consider a pair $(h,\pv{P}{s})$ and a scheduler $\sch$. We qualify
        $\sch$ as \emph{non-blocking w.r.t. $h\pv{P}{s}$} if for every natural
        number $n \in \Nats$ we have some valuation $\mu$ such that $h\pv{P}{s}
        \Downarrow^{\sch,n} \mu$.
\end{definition}
We simply say that $\sch$ is non-blocking  -- \ie\ we omit the reference to
$h\pv{P}{s}$ -- if the pair we are referring to is clear from the context.

Observe that according to our previous remarks every formula $\varphi$
corresponds to an open set of $\PP (S)$ -- in fact we will treat $\varphi$ as
so -- and define inductively the relation $\models$ between $\pv{P}{s}$ and
formulae $\phi$ as follows:
\begin{align*}
        \pv{P}{s} & \models \demon \varphi
        \text{ iff }
        \text{\emph{for all} non-blocking schedulers } \sch  \,
        \text{w.r.t.}\, \pv{P}{s} .  \,
        \exists n \in \Nats, 
        \mu \in \varphi.  \, \pv{P}{s} \Downarrow^{\sch,n} \mu
        \\
        \pv{P}{s} & \models \angel \varphi
        \text{ iff }
        \text{\emph{for some} scheduler } \sch. \>
        \exists n \in \Nats, \,
        \mu \in \varphi. \,
        \pv{P}{s} \Downarrow^{\sch,n} \mu 
\end{align*}
with all other cases defined standardly. The non-blocking condition is
needed to ensure that degenerate schedulers do not trivially falsify
formulae of the type $\demon \varphi$. For example the totally
undefined scheduler always falsifies such formulae.

\section{Denotational semantics}
\label{sec:den}
\textbf{The intensional step.}
As alluded to in Section~\ref{sec:intro} the first step in defining our
denotational semantics is to solve a domain equation of resumptions, with the
branching structure given by a mixed powerdomain (Section~\ref{sec:back}).
More specifically we solve,
\begin{align}
        X \cong \Pow_{x} \PP(S + X \times S)^S
        \hspace{3cm}
        (S \text{ a discrete domain})
        \label{domeq}
\end{align}
for a chosen $x \in \{l,u,b\}$ and where $(-)^S$ is the $S$-indexed product
construct. In order to not overburden notation we abbreviate $\Pow_x$ simply to
$\Pow$ whenever the choice of $x$ is unconstrained. We also denote the functor
$\Pow_x \PP(S + (-) \times S)^S$ by $R_x$ and abbreviate the latter simply to
$R$ whenever the choice of $x$ is unconstrained. Now, in the case that $x \in
\{l,u\}$ the solution of Equation~\eqref{domeq} is obtained standardly, more
specifically by employing the standard final coalgebra construction,
\cite[Theorem IV-5.5]{gierz03}, and the following straightforward proposition.
\begin{proposition}
        \label{theo:cont}
        The functor $R_x$ is locally continuous for every $x \in \{l,u,b\}$.
\end{proposition}
The case $x \in \{b\}$ (which moves us from $\Dom$ to the category $\Coh$) is
just slightly more complex: one additionally appeals to \cite[Exercise
IV-4.15]{gierz03}.  Notably in all cases the solution thus obtained is always
the greatest one \ie\ the carrier $\nu R$ of the final $R$-coalgebra. It is
also the smallest one in a certain technical sense: it is the initial
$R$-strict algebra~\cite[Chapter IV]{gierz03}, a fact that follows from $R$
being strict (\ie\ it preserves strict morphisms) and \cite[Theorem
IV-4.5]{gierz03}. Both finality and initiality are crucial to our work: we will
use the final coalgebra property to define the intensional semantics $\asem{-}$
and then extensionally collapse the latter by recurring to the initial algebra
property.

Let us thus proceed by detailing the intensional semantics.  It is based on the
notion of primitive corecursion~\cite{uustalu99} which we briefly recall next.  

\begin{theorem}
Let $\catC$ be a category with binary coproducts and $\funF :
\catC \to \catC$ be a functor with a final coalgebra $\unfold : \nu \funF \cong
\funF(\nu \funF)$. For every $\catC$-morphism $f : X \to \funF(\nu \funF + X)$
there exists a unique $\catC$-morphism $h : X \to \nu \funF$ that makes the
following diagram commute.
\[
        \xymatrix@C=60pt{
                X \ar[d]_(0.5){f} 
                \ar@{.>}[r]^(0.5){h} 
                & 
                \nu \funF 
                \ar@{}[d]|-(0.5){\cong}
                \ar@/^10pt/[d]^(0.5){\unfold}
                \\
                \funF(\nu F + X) 
                \ar[r]_{\funF[\id,h]}
                &  
                \funF(\nu \funF)
                \ar@/^10pt/[u]^(0.5){\fold}
        }
\]
The morphism $h$ is defined as the composition $\mathrm{corec}([\funF \inl \,
\comp \unfold,f])\comp \inr$, where $\mathrm{corec}([\funF \inl \, \comp
\unfold,f])$ is the universal morphism induced by the $\funF$-coalgebra $[
\funF \inl \, \comp \unfold, f] : \nu F + X \to \funF(\nu \funF + X)$. 
\end{theorem}
With the notion of primitive corecursion at hand one easily defines operators
to interpret both sequential and parallel composition, as detailed next.

\begin{definition}[Sequential composition]
        Define the map $\seqI : \nu R \times \nu R \to \nu R$ as the morphism
        that is induced by primitive corecursion and the composition of
        the following continuous maps,
        \begin{align*}
                \nu R \times \nu R 
                & \xrightarrow{\unfold \times \id} 
                \Pow \PP (S + \nu R \times S)^S \times \nu R \\
                & \xrightarrow{\langle \pi_s \times \id \rangle_{s \in S} } 
                \left ( \Pow \PP (S + \nu R \times S) \times \nu R \right )^S 
                \\
                & \xrightarrow{\mathrm{str}^S}
                \Pow \PP ((S + \nu R \times S) \times \nu R)^S 
                \\
                & \xrightarrow{\cong}
                \Pow \PP ((\nu R + \nu R \times \nu R) \times S)^S 
                \\
                & \xrightarrow{\Pow \PP(\inr)^S}
                \Pow \PP (S + (\nu R + \nu R \times \nu R) \times S)^S 
        \end{align*}
        We denote this composite by $\mathop{\mathrm{seql}}$.
\end{definition}
Intuitively the operation $\mathop{\mathrm{seql}}$ unfolds the resumption on
the left and performs a certain action depending on whether this resumption
halts or resumes after receiving an input: if it halts then
$\mathop{\mathrm{seql}}$ yields control to the resumption on the right; if it
resumes $\mathop{\mathrm{seql}}$ simply attaches the resumption on the right to
the respective continuation.  Note that there exists an analogous operation
$\mathrm{seqr}$ which starts by unfolding the resumption on the right.  Note as
well that $\seqI$ is continuous by construction.

\begin{definition}[Parallel composition]
        Define the map $\parI : \nu R \times \nu R \to \nu R$
        as the morphism that is induced by primitive corecursion and the
        continuous map $\uplus^S \comp \langle \pi_s \times \pi_s \rangle_{s \in S} \comp
        \pv{\mathop{\mathrm{seql}}}{\mathop{\mathrm{seqr}}}$.
\end{definition}
We finally present our intensional semantics $\asem{-}$. It is defined
inductively on the syntactic structure of programs $P \in \mathrm{Pr}$ and
assigns to each program  $P$ a denotation $\asem{P}\in \nu R$. For simplicity
we will often treat the map $\mathrm{unfold}$ as the identity. Note that every
$\asem{P}(s)$ (for $s \in S$) is thus an element of a cone with a
semi-lattice operation $\uplus$ and that this algebraic structure extends to
$\nu R$ by pointwise extension. The denotational semantics is presented in
Figure~\ref{fig:sem}. The isomorphisms used in the last two equations denote
the distribution of products over coproducts, more precisely they denote $S
\times 2 \cong S + S$ with the (coherent) domain $2$ defined as the coproduct
$1 + 1$.  These two equations also capitalise on the bijective correspondence
between elements of $\nu R$ and maps $S \to \Pow \PP (S + \nu R \times S)$.
Note as well that in the case of conditionals we prefix the execution of
$\asem{P}$ (and $\asem{Q}$) with $\asem{\prog{skip}}$.  This is to raise the
possibility of the environment altering states between the evaluation of
$\prog{b}$ and carrying on with the respective branch. An analogous approach is
applied to the case of while-loops. Finally the fact that the map from which we
take the least fixpoint $(\mathrm{lfp})$ is continuous follows from $\seqI$,
copairing, and pre-composition being continuous.

\begin{figure}
        \begin{align*}
                \asem{\prog{skip}} & = s \mapsto x(\{ 1 \cdot s \})
                \\
                \asem{\prog{a}} & = s \mapsto x(\{ \sem{\prog{a}}(s)\})
                \\
                \asem{P ; Q} & = \asem{P} \seqI \, \asem{Q}
                \\
                \asem{P \parallel Q} & = \asem{P} \parI \, \asem{Q}
                \\
                \asem{P +_\prog{p} Q} & = \prog{p} \cdot \asem{P} + 
                (1-\prog{p}) \cdot \asem{Q}
                \\
                \asem{P + Q} & = \asem{P} \uplus
                \asem{Q}
                \\
                \asem{\prog{if} \, \prog{b} \, \prog{then} \, P \, \prog{else} \, Q}                             & = 
                [\asem{\prog{skip}} \seqI \, \asem{Q}, 
                \asem{\prog{skip}} \seqI \, \asem{P}] \, \comp \cong  
                \comp \, \pv{\id}{\sem{\prog{b}}}
                \\
                \asem{\prog{while} \> \prog{b} \> P} & =
                \mathop{\mathrm{lfp}} \Big (r \mapsto [\asem{\prog{skip}},
                        \asem{\prog{skip}} \seqI \left (\asem{P} \seqI\, r \right )] 
                        \, \comp \cong \comp \, \pv{\id}{\sem{\prog{b}}} \Big )
        \end{align*}
        \caption{Intensional semantics}
        \label{fig:sem}
\end{figure}

The denotational semantics thus defined may seem complex, but it actually has a
quite simple characterisation that involves the small-step operational
semantics.  This is given in the following theorem (and corresponds to the
commutativity of the left rectangle in Diagram~\ref{eq:diag}).

\begin{theorem}
        \label{theo:sem_eq} 
        For every program $P \in \mathrm{Pr}$ the denotation $\asem{P}$ is
        (up-to isomorphism) equal to the map, 
        \[
                s \mapsto  
                x \left ( \left \{
                                \textstyle{\sum_i}\, p_i 
                                \cdot \pv{\asem{P_i}}{s_i} + \textstyle{\sum_j}\, p_j \cdot s_j
                                \mid
                                \pv{P}{s} \longrightarrow
                                \textstyle{\sum_i}\, p_i 
                                \cdot \pv{P_i}{s_i} + \textstyle{\sum_j}\, p_j \cdot s_j
                        \right \} \right )
                \]
\end{theorem}
\begin{proof}
        Follows from straightforward induction over the syntactic structure of
        programs, where for the case of while-loops we resort to the fixpoint
        equation.
\end{proof}

\smallskip
\noindent
\textbf{The extensional collapse.} We now present the extensional collapse of
the intensional semantics. Intuitively given a program $P$ the collapse removes
all intermediate computational steps of $\asem{P}$, which as mentioned
previously is a `resumptions tree'.  Technically the collapse makes crucial use
of the fact that the domain $\nu R$ is not only the final coalgebra of
$R$-coalgebras but also the initial algebra of strict $R$-algebras.  More
specifically we define the extensional collapse as the initial strict algebra
morphism,
\[
        \xymatrix@C=80pt{
                \nu R \ar@{.>}[r]^(0.45){\mathrm{ext}} 
                \ar@{}[d]|-(0.5){\cong}
                \ar@/^10pt/[d]^(0.50){\unfold}
                & \Pow \PP (S)^S 
                \\
                \Pow \PP (S + \nu R \times S)^S
                \ar[r]_(0.45){\Pow \PP(\id + \mathrm{ext} \times \id )^S}
                \ar@/^10pt/[u]^(0.50){\fold}
                & 
                \ar[u]_(0.50){{[\eta,\app]^\star}^S} 
                \Pow \PP (S + \Pow \PP(S)^S \times S)^S
        }
\]
Concretely the extensional collapse is defined by $\mathrm{ext}(\asem{P}) =
\bigvee_{n \in \Nats} \, \mathrm{ext}_n(\asem{P}_n)$ where $\asem{P}_n = \pi_n
(\asem{P})$ with $\pi_n : \nu R \to R^n(1)$ (recall that $\nu R$ is a certain
limit with projections $\pi_n : \nu R \to R^n(1)$). The maps $\mathrm{ext}_n$
on the other hand are defined inductively over the natural numbers (and by
recurring to the monad laws) by,
\begin{align*}
        \mathrm{ext}_0 & = \bot \mapsto (s \mapsto x(\{\bot\})) : R^0(1) \to \Pow \PP(S)^S
                \\
        \mathrm{ext}_{n+1} & =
        {[\eta,\app \comp (\mathrm{ext}_n \times \id)]^\star}^S  : R^{n+1} (1) \to \Pow \PP(S)^S
\end{align*}
The map $\mathrm{ext}$ enjoys several useful properties: it is both strict and
continuous (by construction). It is also both linear and $\uplus$-preserving
since it is a supremum of such maps. In the appendix (Section~\ref{sec:seq}) we
show that if the parallel operator is dropped then the semantics obtained from
the composite $\mathrm{ext}\comp \asem{-}$ (hencerforth denoted by $\sem{-}$)
is really just a standard monadic semantics induced by the monad $\Pow \PP$. In
other words, the concurrent semantics obtained from $\mathrm{ext} \comp
\asem{-}$ is a conservative extension of the latter.

We continue unravelling key properties of the composition $\mathrm{ext} \comp
\asem{-}$. More specifically we will now show that for every natural number $n \in
\Nats$, program $P$, and state $s$, the set $\mathrm{ext}_n(\asem{P}_n)(s)$ is
always \emph{finitely generated}, \ie\ of the form $x(F)$ for a certain
\emph{finite} set $F$. We will use this property to prove a proposition which
formally connects the operational and the extensional semantics $\sem{-}$ w.r.t.
$n$-step evaluations.  Thus, the fact that every
$\mathrm{ext}_n(\asem{P}_n)(s)$ is finitely generated follows by induction over
the natural numbers and by straightforward calculations that yield the
equations,
\begin{align}
        \label{eq:fin}
        \begin{split}
        \mathrm{ext}_0(\asem{P}_0)(s) & = x(\{\bot\})
        \\
        \mathrm{ext}_{n+1}(\asem{P}_{n+1})(s) & = x
        \textstyle{
                \left ( \bigcup 
                        \left \{ \sum_i p_i \cdot F_i + \sum_j p_j \cdot s_j \mid
                        \pv{P}{s} \longrightarrow
                        \sum_i p_i \cdot \pv{P_i}{s_i} + \sum_j p_j \cdot s_j 
                \right \} \right )
        }
\end{split}
\end{align}
with every $F_i$ a finite set that satisfies $\mathrm{ext}_n(\asem{P_i}_n)(s_i)
= x(F_i)$. Note that the previous equations provide an explicit construction of
a finite set $F$ such that $x(F) = \mathrm{ext}_n(\asem{P}_n)(s)$. From this we
obtain the formulation and proof of Proposition~\ref{main:theo}, which we
present next.

\begin{proposition}
        \label{main:theo}
        Take a program $P$, state $s$, and
        natural number $n$. The equation below holds.
        \[
                \conv{F} = \left \{ \mu \mid \pv{P}{s} \Downarrow^{\sch,n} \mu
                \text{ with } \sch \text{ a scheduler}
                \right \}
        \]
\end{proposition}
\begin{proof}
First recall that for a natural number $n \in \Nats$, a program $P$, and a
state $s$, we use the letter $F$ to denote the finite set that generates
$\mathrm{ext}_n(\asem{P}_n)(s)$. Note as well that from~\cite[Lemma 2.8]{tix09}
we obtain,
        \begin{align}
                \begin{split}
                        \conv{F}
                & =\mathrm{conv}
                \textstyle{
                \left ( \bigcup \left \{ \sum_i p_i \cdot F_i + \sum_j p_j \cdot s_j \mid
                        \pv{P}{s} \longrightarrow
                \sum_i p_i \cdot \pv{P_i}{s_i} + \sum_j p_j \cdot s_j \right \}\right ) 
                }
                \\
                &
                = \mathrm{conv}
                \textstyle{
                        \left ( \bigcup \left \{ \sum_i p_i 
                                        \cdot \conv{F_i} + \sum_j p_j \cdot s_j \mid
                        \pv{P}{s} \longrightarrow
                \sum_i p_i \cdot \pv{P_i}{s_i} + \sum_j p_j \cdot s_j \right \}\right ) 
                }
                \label{eq:unf}
                \end{split}
        \end{align}
We crucially resort to this observation for the proof, as detailed next.  As
stated in the proposition's formulation we need to prove that the equation,
        \[
                \conv{F} = \left \{ \mu \mid \pv{P}{s} \Downarrow^{\sch,n} \mu
                \text{ with } \sch \text{ a scheduler}
                \right \}
        \]
        holds. We will show that the two respective inclusions hold,
        starting with the case $\supseteq$.  The proof follows by induction
        over the natural numbers and by strengthening the induction invariant
        to encompass all histories $h$ and not just the empty one. The base
        case is direct.  For the inductive step $n+1$ assume that $h\pv{P}{s}
        \Downarrow^{\mathcal{S},n+1} \mu$ for some valuation $\mu$ and
        scheduler $\mathcal{S}$. Then according to the deductive rules of the
        big-step operational semantics (Figure~\ref{fig:bop_sem}) we obtain the
        following conditions:
        \begin{enumerate}
                \item $\sch(h\pv{P}{s}) = \sum_k p_k \cdot \nu_k$ for 
                        some convex combination $\sum_k p_k \cdot (-)$.
                        Moreover $\forall k. \, \pv{P}{s} \longrightarrow \nu_k$
                        with each valuation $\nu_k$ of the form
                        $\sum_i p_{k,i} \cdot \pv{P_{k,i}}{s_{k,i}} +
                        \sum_j p_{k,j}  \cdot s_{k,j}$;
                \item $\forall k,i. \, h\pv{P}{s}\nu_k\pv{P_{k,i}}{s_{k,i}}
                        \Downarrow^{\sch,n} \mu_{k,i}$
                      for some valuation $\mu_{k,i}$; 
                \item and finally 
                $\mu = \textstyle{
                                        \sum_k p_k  \cdot
                                        \left (\sum_{i} p_{k,i} \cdot \mu_{k,i}  + 
                                        \sum_{j} p_{k,j} \cdot s_{k,j} \right )
                                    }$.
        \end{enumerate}
        It follows from the induction hypothesis and the second condition that
        for all $k,i$ the valuation $\mu_{k,i}$ is in $\conv{F_{k,i}}$, where
        $F_{k,i}$ is the finite set that is inductively built from $P_{k,i}$
        and ${s_{k,i}}$ as previously described. Thus for all $k$
        each valuation $\sum_{i} p_{k,i} \cdot \mu_{k,i} + \sum_{j} p_{k,j}
        \cdot s_{k,j}$ is an element of the set $\sum_{i} p_{k,i} \cdot
        \conv{F_{k,i}} +  \sum_{j} p_{k,j} \cdot s_{k,j}$. Also the latter is
        contained in $\conv{F}$ according to Equations~\eqref{eq:unf} and the
        first condition. The proof then follows from the third condition and
        the fact that $\conv{F}$ is convex. 

        Let us now focus on the inclusion $\subseteq$. The base case is again
        direct. For the inductive step $n+1$ recall Equation~\eqref{eq:unf} and
        take a convex combination $\sum_k p_k \cdot \mu_k$ in the set,
        \begin{equation*}
                \conv{F} = 
                \mathrm{conv} 
                \textstyle{
                        \left ( \bigcup \left \{ \sum_i p_i \cdot \mathrm{conv} F_i
                                        + \sum_j p_j \cdot s_j \mid
                        \pv{P}{s} \longrightarrow
                \sum_i p_i \cdot \pv{P_i}{s_i} + \sum_j p_j \cdot s_j \right \}\right ) }
        \end{equation*}
        We can safely assume that every $\mu_k$ belongs to some set $\sum_i
        p_{k,i} \cdot \conv{F_{k,i}} + \sum_j p_{k,j} \cdot s_{k,j}$ generated
        by a corresponding valuation $\sum_i p_{k,i} \cdot
        \pv{P_{k,i}}{s_{k,i}} + \sum_j p_{k,j} \cdot s_{k,j}= \nu_k \in
        \pv{P}{s} \longrightarrow$. Each set $F_{k,i}$ is inductively
        built from $P_{k,i}$ and $s_{k,i}$ in the way that was previously
        described. Crucially we can
        also safely assume that all valuations $\nu_k$ involved are pairwise
        distinct, by virtue of all sets $\conv{F_{k,i}}$ being convex and by
        recurring to the normalisation of subconvex valuations.  

        Next, by construction every $\mu_k$ is of the form $\sum_{i} p_{k,i}
        \cdot \mu_{k,i}  + \sum_{j} p_{k,j} \cdot s_{k,j}$ with $\mu_{k,i} \in
        \conv{F_{k,i}}$, and by the induction hypothesis we deduce that
        $\forall k,i. \, \pv{P_{k,i}}{s_{k,i}} \Downarrow^{\mathcal{S}_{k,i},n}
        \mu_{k,i}$ for some scheduler $\mathcal{S}_{k,i}$. Our next step is to
        construct a single scheduler $\sch$ from all schedulers $\sch_{k,i}$.
        We set $\sch(\pv{P}{s}) = \sum_k p_k \cdot \nu_k$ and $\forall k,i.\,
        \sch(\pv{P}{s}\nu_kl) = \sch_{k,i} (l)$ for every input $l$ with
        $\pv{P_{k,i}}{s_{k,i}}$ as prefix. We set $\sch$ to be undefined w.r.t.
        all other inputs. One then easily proves by induction over the natural
        numbers and by inspecting the definition of the big-step operational
        semantics (Figure~\ref{fig:bop_sem}) that the equivalence,
        \[
                        \pv{P}{s} \nu_k 
                        \pv{P_{k,i}}{s_{k,i}} 
                        \Downarrow^{\mathcal{S},m}
                        \mu
                        \text{ iff }
                        \pv{ P_{k,i} }{ s_{k,i} } \Downarrow^{\mathcal{S}_{k,i}, m}
                        \mu
        \]
        holds for all $k,i$ and natural numbers $m \in \Nats$. 

        Finally one just needs to apply the definition of the big-step
        operational semantics (Figure~\ref{fig:bop_sem}) to obtain $\pv{P}{s}
        \Downarrow^{\mathcal{S},n+1} \sum_k p_k \cdot \mu_k$ as previously
        claimed.
\end{proof}
To conclude this section, we remark that the denotational semantics just
presented is quite different from the well-known \emph{probabilistic testing
semantics} in process algebra (see a brief survey of such semantics for example
in~\cite[Chapters 4 and 5]{deng15}). For instance while testing semantics are
usually based on an operational semantics, denotational ones, as presented
here, are supposed to be independent. Moreover not only they are different in
nature and formulated in quite orthogonal contexts, they sometimes
disagree on which programs/processes should be equated. For example while in
reference~\cite[Chapter 5]{deng15} the terms $P +_\prog{p} (Q + R)$ and $(P
+_\prog{p} Q) + (P +_\prog{p} Q)$ are deemed non-equivalent~\cite[Example
5.7]{deng15}, in our case we have,
\begin{align*}
        & \asem{P +_\prog{p} (Q + R)} 
        \\
        &
        = \text{\big \{ Semantics definition  \big \}}
        \\
        & \prog{p} \cdot \asem{P} + (1 - \prog{p}) \cdot (\asem{Q + R})
        \\
        &
        = \text{\big \{ Semantics definition  \big \}}
        \\
        &
        \prog{p} \cdot \asem{P} + (1 - \prog{p}) \cdot (  \asem{Q} \uplus
        \asem{R})
        \\
        &
        = \text{\big \{ Equational theory of the mixed powerdomains~\cite{tix09}  \big \}}
        \\
        &
        \prog{p} \cdot \asem{P} + ( (1 - \prog{p}) \cdot  \asem{Q} \uplus
        (1 - \prog{p}) \cdot  \asem{R})
        \\
        &
        = \text{\big \{ Equational theory of the mixed powerdomains~\cite{tix09}  \big \}}
        \\
        &
        ( \prog{p} \cdot \asem{P} +  (1 - \prog{p}) \cdot  \asem{Q} ) \uplus
        ( \prog{p} \cdot \asem{P} +  (1 - \prog{p}) \cdot  \asem{R})
        \\
        &
        = \text{\big \{ Semantics definition  \big \}}
        \\
        &
        ( \asem{P +_\prog{p} Q} \uplus
        ( \asem{P +_\prog{p} R})
        \\
        &
        = \text{\big \{ Semantics definition  \big \}}
        \\
        &
        \asem{(P +_\prog{p} Q) + (P +_\prog{p} R)}
\end{align*}
Of course the equation just established will be operationally justified by our
computational adequacy theorem, which is proved in the following section. 

\section{Computational adequacy}
\label{sec:adeq}

In this section we prove the paper's main result: the semantics $\sem{-}$ in
Section~\ref{sec:den} is adequate w.r.t.  the operational semantics in
Section~\ref{sec:lang}. As mentioned in Section~\ref{sec:intro} the formulation
of adequacy relies on the logic that was presented in Section~\ref{sec:lang}.
Actually it relies on different fragments of it depending on which mixed
powerdomain one adopts: in the lower case (\ie\ angelic or may non-determinism)
formulae containing $\demon \varphi$ are forbidden whilst in the upper case
(\ie\ demonic or must non-determinism) formulae containing $\angel  \varphi$
are forbidden.  The biconvex case does not impose any restriction, \ie\ we have
the logic \emph{verbatim}. For simplicity, we will use $\mathcal{L}_l$,
$\mathcal{L}_u$, and $\mathcal{L}_b$ to denote respectively the lower, upper,
and biconvex fragments of the logic.

Recall from Section~\ref{sec:lang} that we used pairs $\pv{P}{s}$ and the
operational semantics of concurrent pGCL to interpret the logic's formulae.
Recall as well that we treat a formula $\varphi$ as an open subset of the space
$\PP(S)$ with $S$ discrete. As the next step towards adequacy, note that the
elements of $\Pow \PP (S)$ also form an interpretation structure for the logic:
specifically we define a satisfaction relation $\models$ between the elements
$A \in \Pow \PP(S)$ and formulae $\phi$ by,
\begin{align*}
        A  \models \angel \varphi
        \text{ iff }
        \textit{for some } \mu \in A
        \text{ we have } \mu \in \varphi
        \hspace{0.8cm} 
        A  \models \demon \varphi
        \text{ iff }
        \textit{for all } \mu \in A
        \text{ we have } \mu \in \varphi
\end{align*}
with the remaining cases defined standardly. The formulation of
computational adequacy then naturally arises: for every program $P$ and state
$s$, the equivalence below holds for all formulae $\phi$.
\begin{align}
        \pv{P}{s} \models \phi \text{ iff }
        \sem{P}(s) \models \phi
        \label{eq:equ}
\end{align}
Of course depending on which mixed powerdomain one adopts $\phi$'s universe of
quantification will vary according to the corresponding fragment
$\mathcal{L}_x$ ($x \in \{l,u,b\}$).  The remainder of the current section is
devoted to proving Equivalence~\eqref{eq:equ}.  Actually the focus is only on
formulae of the type $\angel \varphi$ and $\demon \varphi$, for one can
subsequently apply straightforward induction (over the formulae's syntactic
structure) to obtain the claimed equivalence for all $\phi$ of the
corresponding fragment. We start with the angelic case.

\begin{theorem}
       \label{prop:angel}
       Let $\Pow_x \PP$ be either the lower convex powerdomain or the
       biconvex variant (\ie\ $x \in \{l,b\}$). Then for every program $P$,
       state $s$, and formula $\angel \varphi$ the equivalence below holds.
       \[
                \pv{P}{s} \models \angel \varphi
                \text{ iff }
                \sem{P}(s) \models \angel \varphi
       \]
\end{theorem}
\begin{proof}
        The left-to-right direction follows from Proposition~\ref{main:theo}
        and the upper-closedness of the open $\angel \varphi$. The
        right-to-left direction uses Proposition~\ref{main:theo}, the
        inaccessibility of the open $\angel \varphi$, the characterisation of
        Scott-closure in domains~\cite[Exercise 5.1.14]{larrecq13}, and the
        inaccessibility of the open $\varphi$.
\end{proof}
As already mentioned in the introduction, Equivalence~\eqref{eq:equ} w.r.t.
formulae of the type $\demon \varphi$ is much thornier to prove. In order to
achieve it we will need the following somewhat surprising result.
\begin{theorem}
        \label{prop:alg}
        Consider a program $P$, a state $s$, and a formula $\demon \varphi$.
        If $\pv{P}{s} \models \demon \varphi$ then there exists a positive
        natural number $\mathbf{z} \in \Nats_+$ such that \emph{for all}
        non-blocking schedulers $\sch$ we have,
        \[
                \pv{P}{s} \Downarrow^{\sch, \mathbf{z}}
                \mu \text{ for some valuation $\mu$ and } \mu \in \varphi
        \]
\end{theorem}
In other words whenever $\pv{P}{s} \models \demon \varphi$ there exists an
upper bound -- \ie\ a natural number $\mathbf{z} \in \Nats_+$ such that
\emph{all} non-blocking schedulers (which are \emph{uncountably} many) reach
condition $\varphi$ in \emph{at most} $\mathbf{z}$-steps. Such a property
becomes perhaps less surprising when one recalls K\"onig's
lemma~\cite{franchella97}. The latter in a particular form states that every
finitely-branching tree with each path of finite length must have \emph{finite}
depth (which means that there is an upper-bound on the length of all paths).
From this perspective each non-blocking scheduler intuitively corresponds to a
path and the fact that each path has finite length corresponds to the
assumption that each non-blocking scheduler reaches condition $\varphi$ in a
finite number of steps (\ie\ $\pv{P}{s} \models \demon \varphi$). 

Our proof of Theorem~\ref{prop:alg} is based on a topological generalisation of
the previous analogy to K\"onig's lemma, in which among other things the
condition `finitely-branching' is generalised to `compactly-branching' and the
notion of a path is converted to that of a `probabilistic trace' generated by a
scheduler.  The technical details of this proof and a series of auxiliary
results can be consulted in the appendix (Section~\ref{sec:konig}). Finally by
appealing to Theorem~\ref{prop:alg} we obtain the desired equivalence and
establish computational adequacy.
\begin{theorem}
       \label{prop:demon}
       Let $\Pow_x \PP$ be the upper convex powerdomain or the biconvex
       variant (\ie\ $x \in \{u,b\}$). Then for every program $P$, state $s$,
       and formula $\demon \varphi$ the equivalence below holds.
       \[
                \pv{P}{s} \models \demon \varphi
                \text{ iff }
                \sem{P}(s) \models \demon \varphi
       \]
\end{theorem}
\begin{proof}
        The left-to-right direction follows from Theorem~\ref{prop:alg},
        Proposition~\ref{main:theo} and the upper-closedness of the opens $\varphi$
        and $\demon \varphi$. The right-to-left direction follows from the
        inaccessibility of the open $\demon \varphi$ and
        Proposition~\ref{main:theo}.
\end{proof}

We can now also derive the full abstraction result mentioned in
Section~\ref{sec:intro}. Specifically given programs $P$, $Q$, and state $s$
the observational preorder $\lesssim_x$ is defined by,
\[
        \pv{P}{s} \lesssim_x \pv{Q}{s} \text{ iff }
        \Big ( \forall \phi \in \mathcal{L}_x. \,
                \pv{P}{s} \models \phi \text{ implies }
                \pv{Q}{s} \models \phi \Big )
\]
Then by taking advantage of computational adequacy we achieve full
abstraction.
\begin{corollary}
        \label{theo:adeq}
        Choose one of the three mixed powerdomains $\Pow_x \PP$ $(x \in
        \{l,u,b\})$, a program $P$, a program $Q$, and a state $s$. The
        following equivalence holds.
        \[
                \pv{P}{s} \lesssim_x \pv{Q}{s}  \text{ iff }
                \sem{P}(s) \leq_x \sem{Q}(s)
        \]
\end{corollary}
\begin{proof}
        Follows directly from Theorem~\ref{prop:angel},
        Theorem~\ref{prop:demon}, \cite[Proposition 4.2.4
        and Proposition 4.2.18]{larrecq13}.
\end{proof}

\section{Concluding notes: semi-decidability, quantum, and future work}
\label{sec:conc}

\noindent
\textbf{Semi-decidability.}
The tradition of finding denotational counterparts to operational aspects of
programming languages is well-known and typically very fruitful (see for
example~\cite{winskel93,reynolds98}). Our work does not deviate from this
tradition and introduces the two following families of equivalences,
\[
        \pv{P}{s} \models \angel \varphi
        \text{ iff }
        \sem{P}(s) \models \angel \varphi
        \hspace{2.5cm}
        \pv{P}{s} \models \demon \varphi
        \text{ iff }
        \sem{P}(s) \models \demon \varphi
\]
for appropriate choices of mixed powerdomains $\Pow_x \PP$ ($x \in \{l,u,b\}$).
We also saw that from these it follows a full abstraction theorem w.r.t. the
observational preorder $\lesssim$ described in Section~\ref{sec:intro}. Such
results thus equip our concurrent language with a rich collection of
domain-theoretic tools that one can appeal to in the analysis of several of its
aspects. We briefly illustrate this point next with one example.

Statistical termination is a topic that has been extensively studied over the
years in probability theory~\cite{hart83,lengal17,majumdar24}.  It started
gaining traction recently in the quantum setting as well~\cite{fu24}.  Here we
prove semi-decidability w.r.t. a main class of may and must statistical
termination, an apparently surprising result for it involves quantifications
over the uncountable set of probabilistic schedulers.

We first need to introduce some basic assumptions due to concurrent pGCL being
parametrised.  Take a program $P$ and a state $s$.  We assume that for every
natural number $n \in \Nats$ one can always compute via
Equations~\eqref{eq:fin} the finite set $F_n$ that generates
$\mathrm{ext}_n(\asem{P}_n)(s)$ (recall the end of Section~\ref{sec:den}).
Specifically we assume that every interpretation $\sem{\prog{a}}$ of an atomic
program is computable and only returns linear combinations whose scalars are
rational numbers. Similarly all interpretations $\sem{\prog{b}}$ of conditions
must be decidable. Finally given a subset $U \subseteq S$ we assume that the
membership function $\in_U \,: S \to \{\mathtt{tt},\mathtt{ff}\}$ w.r.t. $U$ is
decidable.

We start with may statistical termination as given by the
formula $\angel \statt U$ with $p$  any number in the set $[0,1) \cap
\Rats$. For this case we fix the mixed powerdomain in our denotational
semantics to be the lower convex one. Next, recall that the statement $\pv{P}{s}
\models \angel \statt U$ refers to the existence of a scheduler
under which $\pv{P}{s}$ terminates in $U$ with probability strictly greater
than $p$. Our algorithmic procedure for checking whether such a
statement holds is to exhaustively search in the finite sets $F_1,F_2,\dots$
for a valuation $\mu$ that satisfies $\statt U$.  Such a procedure is
computable due to the assumptions above, and in order to prove 
semi-decidability we show next that it eventually
terminates whenever $\pv{P}{s} \models \angel \statt U$ holds.

Let us thus assume that $\pv{P}{s} \models \angel \statt U$ holds.  It follows
from Theorem~\ref{prop:angel} and the fact that $\angel \statt U$ corresponds
to an open set of the lower convex powerdomain that there exists a natural
number $n \in \Nats$ such that $\mathrm{ext}_n(\asem{P}_n)(s) \models \angel
\statt U$. In other words the set $\overline{\conv{F_n}}$ contains a valuation
that satisfies $\statt U$ and by the characterisation of Scott-closure in
domains~\cite[Exercise 5.1.14]{larrecq13} there also exists a valuation $\mu
\in \conv{F_n}$ that satisfies $\statt U$.  We will show by contradiction that
such a valuation exists in $F_n$ as well which proves our claim. By the
definition of convex closure we know that $\mu = \sum_{i \in I} p_i \cdot
\mu_i$ for some finite convex combination of valuations in $F_n$. Suppose that
none of the valuations $\mu_i$ $(i \in I)$ satisfy $\mu_i(U) > p$. From this
supposition we take the valuation $\mu_i$ ($i \in I$) with the \emph{largest}
value $\mu_i(U)$, denote it by $\rho$, and clearly $\rho(U) \leq p$.  This
entails $\sum_i p_i \cdot \rho(U) \leq p$ but $\sum_i p_i \cdot \rho(U) \geq
\sum_i p_i \cdot \mu_i(U) > p$, a contradiction.

We now focus on must statistical termination as given by $\demon \statt U$ with
$p \in [0,1) \cap \Rats$. For this case we fix the mixed powerdomain in
our denotational semantics to be the upper convex one. Recall that the
statement $\pv{P}{s} \models \demon \statt U$ asserts that $\pv{P}{s}$
terminates in $U$ with probability strictly greater than $p$ under all
non-blocking schedulers. Our algorithmic procedure for checking whether such a
statement holds is in some sense dual to the previous one: we exhaustively
search for a finite set $F_1,F_2,\dots$ in which all valuations satisfy $\statt
U$. As before such a procedure is computable, and to achieve semi-decidability
we will prove that it eventually terminates whenever $\pv{P}{s} \models \demon
\statt U$ holds.  Thus assume that the latter holds.  It follows from
Theorem~\ref{prop:demon} and the fact that $\demon \statt U$ corresponds to
an open set of the upper convex powerdomain that there exists a natural number
$n \in \Nats$ such that $\mathrm{ext}_n(\asem{P}_n)(s) \models \demon \statt
U$. All valuations in $\mathord{\uparrow}\conv{F_n}$ thus satisfy $\statt U$
and therefore all valuations in $F_n$ satisfy $\statt U$ as well.  

Our reasoning for semi-decidability w.r.t. may statistical termination scales
up to formulae of the type $\angel (\statto{1} U_1 \vee \cdots \vee \statto{n}
U_n)$ with all numbers $p_i \in [0,1) \cap \Rats$.  Dually our reasoning
w.r.t the must variant extends to formulae of the type $\demon (\statto{1} U_1
\wedge \cdots \wedge \statto{n} U_n)$ with all numbers $p_i \in [0,1)
\cap \Rats$.

\newcommand{\bra}[1]{\langle#1|} 
\newcommand{\ket}[1]{|#1\rangle} 
\newcommand{\braket}[2]{ \langle #1 | #2 \rangle} 
\newcommand{\ketbra}[2]{ | #1 \rangle \! \langle #2 |} 

\smallskip
\noindent
\textbf{An application to concurrent quantum computation}.
%
We now apply our results to quantum concurrency, specifically we instantiate
the language in Section~\ref{sec:lang} to the quantum
setting~\cite{nielsen02,watrous18} and obtain computational adequacy for free.
Previous works have already introduced denotational semantics to sequential,
non-deterministic quantum programs~\cite{feng23,feng23b}.  Unlike us however
they do not involve powerdomain structures, although they do mention such would
be more elegant.  Also as far as we aware they do not establish any connection
to an operational semantics, and thus our results are novel already for the
sequential fragment.

We first present our concurrent quantum language, and subsequently the
corresponding state space and interpretation of atomic programs.  We assume
that we have at our disposal $n$ bits and $m$ qubits. We use
$\prog{x_1,x_2,\dots,x_n}$ to identify each bit available and analogously for
$\prog{q_1,q_2,\dots,q_m}$ and qubits. As for the atomic programs we postulate
a collection of gate operations $\prog{U(\vec{q})}$ where $\prog{\vec{q}}$ is a
list of qubits; we also have qubit resets $\prog{q_i \leftarrow} \ket{0}$ ($1
\leq \prog{i} \leq m$) and measurements $\prog{M[x_i \leftarrow q_j]}$ of the
$\prog{j}$-th qubit ($1 \leq \prog{j} \leq m$) with the outcome stored in the
$\prog{i}$-th bit ($1 \leq \prog{i} \leq n$).  The conditions $\prog{b}$ are
set as the elements of the free Boolean algebra generated by the equations
$\prog{x_i = 0}$ and $\prog{x_i = 1}$ for $0 \leq \prog{i} \leq n$.

Recall that a (pure) 2-dimensional quantum state is a unit vector $\ket{\psi}$
in $\Complex^2$, usually represented as a density operator $\ketbra{\psi}{\psi}
\in \Complex^{2 \times 2}$. Here we are particularly interested on the
so-called classical-quantum states~\cite{watrous18}.  They take the
form of a convex combination,
\begin{align}
        \label{eq:cqs}
        \textstyle{\sum_i p_i} \cdot \ketbra{x_i}{x_i} \otimes \ketbra{\psi_i}{\psi_i}
\end{align}
where each $x_i$ is an element of $2^n$ (\ie\ a classical state) and each
$\ket{\psi_i}$ is a pure quantum state in $\Complex^{2^{\otimes m}}$. Such
elements are thus distributions of $n$-bit states $\ketbra{x_i}{x_i}$ paired
with $m$-qubit states $\ketbra{\psi_i}{\psi_i}$. The state space $S$ that
we adopt is the subset of $\Complex^{{2 \times 2}^{\otimes n}} \otimes
\Complex^{{2 \times 2}^{\otimes m}}$ that contains precisely the elements of
the form $1 \cdot \ketbra{x}{x} \otimes \ketbra{\psi}{\psi}$ as previously
described -- \ie\ we know with certainty that $\ketbra{x}{x} \otimes
\ketbra{\psi}{\psi}$ is the current state of our classical-quantum system. We
call such states \emph{pure} classical-quantum states.

We then interpret atomic programs as maps $S \to \Dist(S)$ which as indicated
by their signature send pure classical-quantum states into arbitrary
ones~\eqref{eq:cqs}.  More technically we interpret such programs as
restrictions of completely positive trace-preserving operators, a standard
approach in the field of quantum information~\cite{watrous18}. Completely
positive trace-preserving operators are often called quantum channels and we
will use this terminology throughout the section. 

We will need a few preliminaries for interpreting the atomic programs. First a
very useful quantum channel is the trace operation $\mathrm{Tr} : \Complex^{{2
\times 2}^{\otimes n}} \to \Complex^{{2 \times 2}^{\otimes 0}} = \Complex$
which returns the trace of a given matrix~\cite[Corollary 2.19]{watrous18}.
Among other things it induces the partial trace on $m$-qubits,
\[
        \mathrm{Tr}_i := (\otimes^{i -1}_{j = 1} \id) \otimes \mathrm{Tr} \otimes 
        (\otimes^{m}_{j = i + 1}\, \id)
\]
which operationally speaking discards the $i$-th qubit. It is also easy to see
that for every density operator $\rho \in \Complex^{{2\times 2}^{\otimes n}}$
the map $\rho : \Complex^{{2\times 2}^{\otimes 0}} \to \Complex^{{2\times
2}^{\otimes n}}$ defined by $1 \mapsto \rho$ is a quantum
channel~\cite[Proposition 2.17]{watrous18}. Finally note that one can always
switch the positions $i$ and $j$ of two qubits~\cite[Corollary 2.27]{watrous18}
in a pure quantum state.

We now present the interpretation of atomic programs. First for each gate
operation $\prog{U(\vec{q})}$ we postulate the existence of a unitary operator
$U : \Complex^{2^{\otimes m}} \to \Complex^{2^{\otimes m}}$. Then we interpret
gate operations and resets by,
\begin{align*}
        \sem{\prog{U(\vec{q})}}\big (\ketbra{x}{x} \otimes \ketbra{\psi}{\psi} \big)
        & = \ketbra{x}{x} \otimes U\ketbra{\psi}{\psi}U^\dagger
        \\
        \sem{\prog{q_i} \leftarrow \ket{0}} 
        \big (\ketbra{x}{x} \otimes \ketbra{\psi}{\psi} \big)
        & = \ketbra{x}{x} \otimes \mathrm{mov_i} \big (\ketbra{0}{0} \otimes 
        \mathrm{Tr_i}\ketbra{\psi}{\psi} \big )
\end{align*}
where $U^\dagger$ is the adjoint of $U$ and $\mathrm{mov_i}$ is the operator
that moves the leftmost state (in this case $\ketbra{0}{0}$) to the
$\prog{i}$-th position in the list of qubits. In order to interpret
measurements we will need a few extra auxiliary operators. First we take the
`quantum-to-classical' channel $\Phi : \Complex^{2 \times 2} \to \Complex^{{2
\times 2}^{\otimes 2}}$ defined by,
\[
        \begin{pmatrix}
                a && b 
                \\
                c && d
        \end{pmatrix}
        \mapsto
        a \cdot \ketbra{0}{0}^{\otimes 2}
        + 
        d \cdot \ketbra{1}{1}^{\otimes 2}
\]
The fact that it is indeed a channel follows from~\cite[Theorem
2.37]{watrous18}. In words $\Phi(\rho)$ is a measurement of $\rho$ w.r.t.  the
computational basis: the outcome $\ket{0}$ is obtained with probability $a$ and
analogously for $\ket{1}$. In case we measure $\ket{0}$ we return
$\ketbra{0}{0}^{\otimes 2}$ and analogously for $\ket{1}$. Note the use of
$\ketbra{0}{0}^{\otimes 2}$ (and not $\ketbra{0}{0}$) so that we can later
store a copy of $\ketbra{0}{0}$ in the classical register. To keep the notation
easy to read we use $\Phi_i$ to denote the operator that applies $\Phi$ in the
$i$-th position and the identity everywhere else. Finally we have,
\[
        \sem{\prog{M[x_i \leftarrow q_j}]} ( \ketbra{x}{x} \otimes \ketbra{\psi}{\psi} )
        = 
        \mathrm{mov}^{\mathrm{cq}}_\prog{j,i} \left (\mathrm{Tr_\prog{i}}\ketbra{x}{x} 
                \otimes \Phi_\prog{j}\ketbra{\psi}{\psi}
        \right )
\]
where $\mathrm{mov}^\mathrm{cq}_\prog{j,i}$ sends the $\prog{j}$-th qubit to
the $\prog{i}$-th position in the list of bits. 


\bigskip
\noindent
\textbf{Future work}.
We plan to explore a number of research lines that stem directly from our work.
For example we would like to expand our (brief) study about semi-decidability
w.r.t. $\pv{P}{s} \models \phi$. More specifically we would like to determine
the largest class of formulae $\phi$ in our logic (recall
Section~\ref{sec:lang}) under which one can prove semi-decidability.  Second it
is well-known that under mild conditions the three mixed powerdomains are
isomorphic to the so-called prevision models~\cite{gobault15}. In other words
via Section~\ref{sec:den} and Section~\ref{sec:adeq} we obtain for free a
connection between our (concurrent) language and yet another set of
domain-theoretic tools.  Will this unexplored connection provide new insights
about the language? It would also be interesting to investigate in what ways a
semantics based on event structures~\cite{varacca06b} complements the one
presented here.

We also plan to extend concurrent pGCL in different directions.  The notion of
fair scheduling for example could be taken into account which leads to
countable non-determinism. We conjecture that the notion of a mixed powerdomain
will need to be adjusted to this new setting, similarly to what already happens
when probabilities are not involved~\cite{apt86}.  Another appealing case is
the extension of our concurrent language and associated results to the
higher-order setting, obtaining a language similar in spirit to concurrent
idealised Algol~\cite{brookes96}. A promising basis for this is the general
theory of PCF combined with algebraic
theories~\cite{plotkin01,plotkin03,plotkin08}. In our case the algebraic theory
adopted would need to be a combination of that of states and that of mixed
non-determinism. 

\bigskip
\noindent
\textbf{Acknowledgements.} This work is financed by National Funds through FCT
- Fundação para a Ciência e a Tecnologia, I.P. (Portuguese Foundation for
Science and Technology) within project IBEX, reference
10.54499/PTDC/CCI-COM/4280/2021
(https://doi.org/10.54499/PTDC/CCI-COM/4280/2021).

\bibliographystyle{eptcs}
\bibliography{biblio}

\pagebreak
\appendix
\section{The sequential fragment}
\label{sec:seq}
Recall that we denote the extensional semantics $\mathrm{ext} \comp \asem{-}$
introduced in the main text by $\sem{-}$. In this section we prove that
$\sem{-}$ satisfies the equations listed in Figure~\ref{fig:comp} whenever the
parallel composition operator is not involved. We start with the following
lemma.
  \begin{figure}
        \begin{flalign*}
                \sem{\prog{skip}} 
                & = s \mapsto x(\{ 1 \cdot s \})
                \\
                \sem{\prog{a}} 
                & = s \mapsto x( \{ \sem{\prog{a}}(s) \} )
                \\
                \sem{P ; Q}
                & = \sem{Q}^\star
                \comp \sem{P}
                \\
                \sem{P +_\prog{p} Q}
                & = \prog{p} \cdot 
                \sem{P}
                + (1-\prog{p}) \cdot \sem{Q}
                \\
                \sem{P + Q}
                & =   
                        \sem{P} \uplus
                        \sem{Q}
                \\
                \sem{\prog{if} \, \prog{b} \, 
                        \prog{then} \, P \, \prog{else} \, Q}
                                          & = 
                \big [\sem{Q}, 
                \sem{P} \big ] \, \comp \cong  
                \comp \, \pv{\id}{\sem{\prog{b}}}
                \\
                \sem{\prog{while} \> \prog{b} \> P} 
                                          & =
                \mathop{\mathrm{lfp}} \Big (f \mapsto \big [ \eta,
                        f^\star \comp \sem{P} \big ] 
                \, \comp \cong \comp \, \pv{\id}{\sem{\prog{b}}} \Big )
        \end{flalign*}
        \caption{Sequential fragment of the extensional semantics}
        \label{fig:comp}
        \end{figure}

\begin{lemma}
        \label{mon:comp}
        The equation below holds for all programs $P \in \Pr$ 
        and elements $r$ of $\nu R$.
        \[      
                \textstyle{\bigvee_{n \in \Nats}}\, (\mathrm{ext}(r))^\star
                \comp \mathrm{ext}_n(\asem{P}_n)
                =
                \mathrm{ext}\,(\asem{P} \seqI r)
        \]
\end{lemma}

\begin{proof}
        We will show that the two respective inequations (\ie\ $\leq_x$ and
        $\geq_x$) hold.  We start with the case $\leq_x$ which we will
        prove by showing that for every natural number $n \in \Nats$ the
        inequation below holds.
        \[
                (\mathrm{ext}(r))^\star
                \comp \mathrm{ext}_n(\asem{P}_n)
                \leq_x
                \mathrm{ext}(\asem{P} \seqI r)
        \]
        We proceed by induction over the natural numbers. The base case is
        direct and for the inductive step we reason as follows,
        \begin{align*}
                        & \> (\mathrm{ext}(r))^\klcomp
                        (\mathrm{ext}_{n+1}(\asem{P}_{n+1})(s))
                        \\[4pt]
                        &
                        \text{\big \{ Defn. of $\mathrm{ext}_{n+1}$
                                        and Theorem~\ref{theo:sem_eq} \big \} }
                           & 
                           \\[4pt]
                        & =
                        (\mathrm{ext}(r))^\klcomp
                        \left (\mathlarger{\uplus} \left 
                        \{ \textstyle{\sum_i p_i} \cdot 
                                \mathrm{ext}_n(\asem{P_i}_n)
                        (s_i) + \textstyle{\sum_j p_j} \cdot \eta(s_j) \mid
                        \dots 
                        \right \}
                        \right )
                        &
                \\[4pt] &
                        \text{\big \{ $(\mathrm{ext}(r))^\klcomp$ is linear and 
                        $\uplus$-preserving \big \}}
                           & 
                           \\[4pt]
                        &
                        =
                        \mathlarger{\uplus} \left 
                        \{ \textstyle{\sum_i p_i} \cdot 
                              (\mathrm{ext}(r))^\klcomp
                              (\mathrm{ext}_n(\asem{P_i}_n)
                        (s_i)) + \textstyle{\sum_j p_j} \cdot 
                        \mathrm{ext}(r)(s_j) 
                        \mid  \dots     \right \}
                        &
                \\[4pt] &
                        \text{\big \{ Induction hypothesis and monotonicity \big \} }
                           &
                           \\[4pt]
                        &
                        \leq_x
                        \mathlarger{\uplus} \left 
                        \{ \textstyle{\sum_i p_i} \cdot 
                                \mathrm{ext}(\asem{P_i} \seqI r)
                        (s_i) + \textstyle{\sum_j p_j} \cdot 
                        \mathrm{ext}(r)(s_j) 
                        \mid 
                        \dots 
                        \right \}
                        &
                        \\ &
                        \text{\big \{ Defn. of $(-)^\klcomp$, $\seqI$, and
                        Theorem~\ref{theo:sem_eq} \big \}}
                           &
                           \\[4pt]
                        &
                        = [\eta,\mathrm{app} \comp 
                        (\mathrm{ext} \times \id)]^\klcomp
                        \left  ( \asem{P} \seqI r (s)
                                \right )
                        &
                \\[4pt] &
                        \text{\big \{ Defn. of $\mathrm{ext}$ \big \}}
                           &
                           \\[4pt]
                        &
                        = \mathrm{ext} (\asem{P} \seqI r)(s)
        \end{align*} 
        where the ellipsis ($\dots$) hides the expression $\pv{P}{s}
        \longrightarrow \textstyle{\sum_i} p_i \cdot \pv{P_i}{s_i} +
        \textstyle{\sum_j} p_j \cdot s_j$. 
        The case $\geq_x$ follows from showing that for every natural number
        $n \in \Nats$ the inequation,
        \[
                (\mathrm{ext}(r))^\star
                \comp \mathrm{ext}_n(\asem{P}_n)
                \geq_x
                \mathrm{ext}_n((\asem{P} \seqI r)_n)
        \]
        holds. This is proved by induction over the natural numbers, and since
        it is completely analogous to the previous case we omit the details.
\end{proof}

\begin{theorem}
        The equations in Figure~\ref{fig:comp} are sound.
\end{theorem}
\begin{proof}
The first two cases are direct. The third case is obtained from
Lemma~\ref{mon:comp} and the fact that the post-composition of Scott-continuous
maps is Scott-continuous.
The fourth and fifth cases follow directly from $\mathrm{ext}$ being both
linear and $\uplus$-preserving. 
%
The case involving conditionals is similar to the case involving while-loops
which we detail next.  

For while-loops we resort to the fact that both intensional and extensional
semantics use Kleene's fixpoint theorem (recall Figure~\ref{fig:sem} and
Figure~\ref{fig:comp}). Specifically they involve mappings $\sigma^n : \nu R
\to \nu R$ (recall that $\nu R$ is treated as a set of functions) and $\tau^n :
[S, \Pow \PP(S)] \to [S, \Pow \PP(S)]$ (for all $n \in \Nats$). These are
defined as follows:
\begin{align*}
        & \begin{cases}
                \sigma^0(r) & = \bot 
                \\
                \sigma^{n+1}(r) & = [\asem{\prog{skip}}, 
                \asem{\prog{skip}} \seqI \, (\asem{P} \seqI \sigma^n(r))] 
                \, \comp \cong \comp \, \pv{\id}{\sem{\prog{b}}}
        \end{cases}
       \\[4pt]
        &  \begin{cases}
                \tau^0(f) & = \bot 
                \\
                \tau^{n+1}(f) & = [\eta, (\tau^n(f))^\star \comp 
                \mathrm{ext}(\asem{P})] \, \comp \cong \comp \, \pv{\id}{\sem{\prog{b}}}
        \end{cases}
\end{align*}
It follows from straightforward induction over the natural numbers that the
equation $\mathrm{ext}(\sigma^{n}(\bot)) = \tau^n(\bot)$ holds for all $n \in \Nats$.
Then we reason:
\begin{align*}
        & \> \, \mathrm{ext}(\asem{\prog{while \, b} \, P}) 
        \\
        & = \text{ \big \{ Defn. of $\asem{-}$ \big \}}
        \\
        &
        \, \mathrm{ext}(\textstyle{\bigvee_n} \, \sigma^n(\bot))
        \\
        &
        = \text{\big \{ Scott-continuity of $\mathrm{ext}$ \big \}}
        \\
        &
        \, \textstyle{\bigvee_n}\,  \mathrm{ext}(\sigma^n(\bot))
        \\ 
        & = \text{\big \{ $\mathrm{ext}(\sigma^n(\bot)) = \tau^n(\bot)$  } \big \}
        \\
        & \, \textstyle{\bigvee_n}\, \tau^n(\bot)
        \\
        & = \text{ \big \{Kleene's lfp construction \big \}}
        \\
        &
        \, \mathrm{lfp} \Big (f \mapsto [\eta, f^\star \comp \mathrm{ext}(\asem{P})]
        \, \comp \cong \comp \, \pv{\id}{\sem{\prog{b}}} \Big )
        \end{align*}
\end{proof}

\section{Proof of Theorem~\ref{prop:alg}}
\label{sec:konig}

In this section we prove Theorem~\ref{prop:alg}, which was presented in the
main text. It requires a series of preliminaries which are detailed and proved
next.

\begin{theorem}
        \label{theo:comp}
        Consider a family of zero-dimensional coherent domains $(X_i)_{i \in
        I}$ and also a finite set $F \subseteq \prod_{i \in I} \PP_{= 1}(X_i)
        \subseteq \prod_{i \in I} \PP(X_i)$. Then $b(F) = \mathrm{conv}\, F$
        and consequently the set $\mathrm{conv}\, F$ is Lawson-compact.
\end{theorem}

\begin{proof}
        Take $(x_i)_{i \in I} \in \overline{\conv{F}} \cap \mathord{\uparrow}\,
        \conv{F}$. We can assume the existence of an element $(a_i)_{i \in I}
        \in \conv{F}$ such that $(x_i)_{i \in I} \geq (a_i)_{i \in I}$. Note
        that every $a_i$ $(i \in I)$ has total mass $1$ by assumption.  Since
        $(x_i)_{i \in I} \in \overline{\conv{F}}$ it must be the case that
        every $x_i$ ($i \in I$) has total mass $\leq 1$~\cite[Exercise
        5.1.14]{larrecq13}, and therefore every $x_i$ must have total mass $1$
        due to the inequation $(x_i)_{i \in I} \geq (a_i)_{i \in I}$. This
        means that both $(x_i)_{i \in I}$ and $(a_i)_{i \in I}$ live in
        $\prod_{i \in I} V_{=1}(X_i)$ and thus $(x_i)_{i \in I} = (a_i)_{i \in
        I}$ (by Theorem~\ref{theo:locin}). It follows that $(x_i)_{i \in I}
        \in \conv{F}$ which proves the claim.
\end{proof}
The next step is to extend the previous theorem to finite sets of tuples of
subvaluations, \ie\ we extend the previous result to finite sets $F \subseteq
\prod_{i \in I} \PP_{\leq 1}(X_i)$ with every $X_i$ a zero-dimensional coherent
domain. We will need the following auxiliary result.

\begin{lemma}
        \label{lem:eppair}
        For every domain $X$  there exists an ep-pair between $\PP(X + 1)$ and
        $\PP(X)$. Specifically the projection $p: \PP(X+1) \twoheadrightarrow
        \PP(X)$ is defined as $[\eta^\PP,\const{\bot}]^\star$ whereas the
        embedding $e: \PP(X) \hookrightarrow \PP(X+1)$ is defined as
        $\PP(\inl)$.
\end{lemma}

\begin{proof}
        The notion of an ep-pair can be consulted for example in~\cite[Section
        IV-1]{gierz03} and also \cite[Definition 9.6.1]{larrecq13}.  The result
        itself follows by straightforward calculations.
\end{proof}

\begin{theorem}
         \label{theo:comp2}
         Consider a family of zero-dimensional coherent domains $(X_i)_{i
         \in I}$ and also a finite set $F \subseteq \prod_{i \in I} \PP_{\leq
         1}(X_i) \subseteq \prod_{i \in I} \PP(X_i)$. Then the set
         $\conv{F}$ is Lawson-compact.
\end{theorem}

\begin{proof}
        It follows from Lemma~\ref{lem:eppair} and the fact that ep-pairs are
        preserved by products that there exists an ep-pair between $\prod_{i
        \in I} \PP(X_i +1)$ and $\prod_{i \in I} \PP(X_i)$. In particular by
        \cite[Exercise O-3.29]{gierz03} we
        have a Lawson-continuous map $p: \prod_{i \in I}
        \PP(X_i+1)\twoheadrightarrow \prod_{i \in I} \PP(X_i)$ defined by
        $\prod_{i \in I} [\eta^\PP,\const{\bot}]^\star$. Note as well that
        every $X_i + 1$ is zero-dimensional~\cite[Proposition
        5.1.59]{larrecq13} and a coherent domain.

        Consider now a tuple $(\mu_i)_{i \in I} \in F \subseteq \prod_{i \in I}
        \PP_{\leq 1}(X_i)$. For every such tuple we define a new one
        $(\mu'_i)_{i \in I} \in \prod_{i \in I} V_{=1} (X_i + 1)$ by setting,
        \[
                \mu_i' = \PP(\inl)(\mu_i) + (1 - \mu_i(X_i)) \, \cdot \,
        {\inr(\ast)}
        \]
        We obtain a new set $F' \subseteq \prod_{i
        \in I} V_{=1} (X_i + 1)$ in this way and by Theorem~\ref{theo:comp} its
        convex closure $\conv{F'}$ is Lawson-compact. Finally by taking
        advantage of the fact that $p$ is linear we reason $p[\conv{F'}] =
        \conv{p[F']} = \conv{F}$, and thus since $p$ is Lawson-continuous the
        set $\conv{F}$ must be Lawson-compact.
\end{proof} 
Together with~\cite[Proposition III-1.6]{gierz03} the previous theorem entails
that under certain conditions the closure $\overline{\conv{F}}$ of the previous
set $\conv{F}$ simplifies to $\mathord{\downarrow}\conv{F}$. This is 
formulated in the following corollary, and it will be useful in the proof of
Theorem~\ref{prop:alg}.

\begin{corollary}
         \label{cor:basis}
         Consider a family of zero-dimensional coherent domains $(X_i)_{i \in
         I}$ and also a finite set $F \subseteq \prod_{i \in I} \PP_{\leq
         1}(X_i) \subseteq \prod_{i \in I} \PP(X_i)$. Then we have
         $\overline{\conv{F}} = \mathord{\downarrow}\conv{F}$ and in particular
         we obtain $b(F) = \mathord{\downarrow}\, \mathrm{conv}\, F \cap
         \mathord{\uparrow}\, \mathrm{conv}\, F$.
\end{corollary}

\smallskip
\noindent
\textbf{The notion of an $n$-step trace.}
Next we introduce the notion of an $n$-step trace. Going back to the analogy to
K\"onig's lemma (recall Section~\ref{sec:adeq}), an $n$-step trace corresponds
to a path of length $n$ in a tree. Such traces involve a space of choices (made
by a scheduler along a computation) which we detail next. In order to not
overburden notation we denote the set of inputs of a scheduler, \ie\
\[
        ((\SPrg \times S)  \times \PP_{=1} (S + \SPrg \times S) )^\ast 
        \times (\SPrg
        \times S)
\]
by $I$. Then we define the space of choices $\Ch$ as the $\Dcpo$-product
$\textstyle{\prod_{i \in I}} \PP \big ( \PP_ {=1}(S + \SPrg \times S) \big)$.
The space $S + \SPrg \times S$ is discrete and thus $\PP_ {=1}(S + \SPrg \times
S)$ is also discrete (Theorem~\ref{theo:locin}), in particular it is a coherent
domain. This implies that $\Ch$ is an $\LCone$-product per our previous remarks
about products in Section~\ref{sec:back}. All schedulers can be uniquely
encoded as elements of $\Ch$. They thus inherit an order which is precisely
that of partial functions by Theorem~\ref{theo:locin}. The inherited order on
schedulers is furthermore a DCPO. Next, given a valuation $\nu \in \PP_{=1}(S +
\SPrg \times S)$ we use the expression $h\pv{P}{s} \mapsto 1 \cdot \nu$ to
denote the element in $\Ch$ that associates the input $h\pv{P}{s}$ to $1 \cdot
\nu$ and all other inputs to the zero-mass valuation $\bot$.  Given a linear
combination $\nu \in \PP(S + \SPrg \times S)$ we use $\nu^S \in \PP(S)$ to
denote the respective restriction on $S$. We now formally introduce the notion
of an $n$-step trace.

\begin{definition}
        \label{defn:trace}
        Take a pair $(h,\pv{P}{s})$ and a scheduler $\sch$ such that
        $\sch(h\pv{P}{s}) = \sum_k p_k \cdot \nu_k$ for some convex
        combination $\sum_k p_k \cdot (-)$ with each $\nu_k$ of the form,
        \[
             \textstyle{   \sum_{i} p_{k,i} \cdot \pv{P_{k,i}}{s_{k,i}} + \sum_{j} p_{k,j}
             \cdot s_{k,j} }
        \]
        Then we define an $\Nats_+$-indexed family of partial
        functions $t^{\sch,n} : I \to (\PP(S) \times \Ch)^n$, which send an
        input to the respective $n$-step trace, as follows:
        \begin{align*}
        t^{\sch,1}(h\pv{P}{s}) & = 
        \textstyle{\sum_k} \, p_k  \cdot \left (\nu^S_k, h\pv{P}{s} \mapsto
        1 \cdot \nu_k \right ) \\
        t^{\sch,{n+1}}(h\pv{P}{s}) & = 
        \textstyle{ \sum_k\, } p_k \cdot \Big ((\nu^S_k, h\pv{P}{s} \mapsto
        1 \cdot \nu_k ) \, \mathtt{::} \,
                                \\ & \hspace{2cm}
        \textstyle {\sum_{i}\, } 
        p_{k,i} \cdot t^{\sch,n}(h\pv{P}{s}\nu_k\pv{P_{k,i}}{s_{k,i}})
        + \Delta^n(\nu^S_k,\bot) \Big )
\end{align*}
The operation $(\mathtt{::})$ is the append function on lists and $\Delta^n :
X \to X^n$ is the $n$-diagonal map which sends an input $x$ to the $n$-tuple
$(x,\dots,x)$.
\end{definition}
Despite looking complicated the previous definition can be easily related to
the big-step operational semantics. We detail this in the following theorem.

\begin{theorem}
        \label{theo:tracef}
        Take a pair $(h,\pv{P}{s})$, a scheduler $\sch$, and a positive natural
        number $n \in \Nats_+$. If $t^{\sch,n}(h\pv{P}{s}) = (\mu_1,c_1)
        \dots (\mu_n,c_n)$ for some sequence $(\mu_1,c_1) \dots (\mu_n,c_n)$
        then for all $1 \leq i \leq n$ we have $h\pv{P}{s}
        \Downarrow^{\sch,i} \mu_i$. Conversely if for all $1 \leq i \leq n$
        we have $h\pv{P}{s} \Downarrow^{\sch,i} \mu_i$ for some sequence of
        valuations $\mu_1 \dots \mu_n$ then we must also have
        $t^{\sch,n}(h\pv{P}{s}) = (\mu_1,c_1)\dots(\mu_n,c_n)$ for some
        sequence of choices $c_1 \dots c_n$.
\end{theorem}
\begin{proof}
        The first claim follows from straightforward induction over the
        positive natural numbers, an unravelling of the trace functions'
        definition (Definition~\ref{defn:trace}), and applications of the
        operational semantics' rules (Figure~\ref{fig:bop_sem}).  The second
        claim follows from induction over the positive natural numbers, the
        first claim and an analysis of both the operational semantics' rules
        and the conditions under which a trace function is well-defined.
\end{proof}
The following result involving traces will also be useful.
\begin{lemma}
        \label{lem:pref}
        The two conditions below hold for every
        pair $(h,\pv{P}{s})$.
        \begin{itemize}
                \item For every pair of schedulers $\sch_1$ and $\sch_2$
                        such that $\sch_1 \leq \sch_2$ the implication
                        below holds for every natural number $n \in \Nats_+$.
                        \[
                                t^{\sch_1,n}(h\pv{P}{s})
                                \text{ is well-defined}
                                \Longrightarrow
                                t^{\sch_1,n}(h\pv{P}{s})
                                =
                                t^{\sch_2,n}(h\pv{P}{s})
                        \]
                \item For every scheduler $\sch$ and positive natural number
                        $n \in \Nats_+$ the trace $t^{\sch,n}(h\pv{P}{s})$ is
                        a prefix of $t^{\sch,n+1}(h\pv{P}{s})$ whenever
                        the latter is well-defined.
        \end{itemize}
\end{lemma}

\begin{proof}
        Both cases follow from simple induction over the positive natural
        numbers.
\end{proof}
Next, the following lines are devoted to showing that the set of \emph{all}
$n$-step traces of $h\pv{P}{s}$ is Lawson-compact. Intuitively this is
connected to the generalisation of the condition `finitely-branching' in
K\"onig's lemma to `compactly-branching'.  In order to prove that the set of
all $n$-step traces $(n \in \Nats_+)$ of $h\pv{P}{s}$ is Lawson-compact we
recur to the biconvex powercone and associated machinery.  In particular we
start by considering the continuous map $\mathrm{step} : I \rightarrow \Pow_b
(\PP(S) \times \Ch \times \PP(S + I))$ defined by,
\[
        h\pv{P}{s} \mapsto b\ \Big \{
        \Big (\nu^S,\, h \pv{P}{s} \mapsto 1 \cdot \nu,  \,
        \PP \big (\id + (x \mapsto h\pv{P}{s}\nu x) \big )(\nu)
\Big ) \mid 
        \nu \in \pv{P}{s} \longrightarrow \Big \}
\]
Next we define a family of maps $(f_n)_{n \in \Nats_+} : I \to \Pow_b((\PP(S)
\times \Ch)^n)$ in $\Coh$ (\ie\ the category of coherent domains and continuous
maps) by induction on the size of $n \in \Nats_+$.  Concretely for the
inductive step we set,
\begin{align*}
        f_{n+1} = I
        & \xrightarrow{\mathrm{step}} \Pow_b \left ( \PP (S) \times \Ch \times
        \PP(S + I) \right ) \\
        & \xrightarrow{\Pow_b \left 
        (\id \times [\mathrm{stop},f_n]^{\star^\PP} \right ) }
        \Pow_b \left ( \PP(S) \times \Ch \times \Pow_b ((\PP(S) \times \Ch)^n) 
        \right )  
        \\
        & \xrightarrow{\mathrm{str}^{\star^{\Pow_b}}}
        \Pow_b \left ((\PP(S)\times \Ch)^{n+1} \right )
\end{align*}
and set $f_1 = \Pow_b (\pi_1) \comp \mathrm{step}$ for the base step, with
$\mathrm{stop} : S \to \Pow_b((\PP(S) \times \Ch)^n)$ as the composite
$\eta^{\Pow_b} \comp \Delta^n \comp \pv{\eta^\PP}{\const{\bot}}$.  Note that
every space $f_n(h\pv{P}{s}) \subseteq (\PP(S) \times \Ch)^n$ is Lawson-compact
by construction. Our next step is to prove that $f_n(h\pv{P}{s})$ is
actually the set of all $n$-step traces of $h\pv{P}{s}$.  First, the following
equations are obtained from straightforward calculations,
\begin{align*}
                f_1(h\pv{P}{s}) & = b  
                \left \{ (\nu^S, h \pv{P}{s} \mapsto 1 \cdot \nu )
                \mid \nu \in \pv{P}{s} \longrightarrow \right \}
                \\
                f_{n+1}(h\pv{P}{s}) & = b \left ( \bigcup_{
                        \nu \in \pv{P}{s} \longrightarrow}
                       \left \{ \left (\nu^S, h\pv{P}{s} \mapsto 1 \cdot \nu \right)
                        \right \} \times
                        \left (\textstyle{\sum_i}\, p_i \cdot \,
                                 F_i 
                        + \Delta^n(\nu^S,\bot) \right ) \right )
\end{align*}
under the assumption that every $F_i$ is a finite set that satisfies the
equation $b(F_i) = f_n(h\pv{P}{s}\nu\pv{P_i}{s_i})$ and that all valuations
$\nu \in \pv{P}{s} \longrightarrow$ are of the form $\sum_i p_i \cdot
\pv{P_i}{s_i} + \sum_j p_j \cdot s_j$. It then follows by induction over the
positive natural numbers that every $f_n(h\pv{P}{s})$ is of the form $b(F)$ for
a finite set $F$ that is built inductively by resorting to the two previous
equations. We will use the expression $F_n$ to denote such a set w.r.t.
$f_n(h\pv{P}{s})$ unless stated otherwise.  Sometimes we will abbreviate $F_n$
simply to $F$ if $n$ is clear from the context. Second by an appeal to
Corollary~\ref{cor:basis} we obtain $f_n(h\pv{P}{s}) = b(F) =
\mathord{\downarrow} \conv{F} \cap \mathord{\uparrow} \conv{F}$ which further
simplifies our set description of $f_n(h\pv{P}{s})$. Then we concretely connect
the spaces $f_n(h\pv{P}{s})$ to the trace functions
(Definition~\ref{defn:trace}) via the two following lemmata.

\begin{lemma}
        \label{lem:inc}
        Consider a program $P$, a state $s$, an history $h$, and a positive
        natural number $n \in \Nats_+$. For every scheduler $\sch$ we have
        $t^{\sch,n}(h\pv{P}{s}) \in \conv{F} \subseteq f_n(h\pv{P}{s})$
        whenever $t^{\sch,n}(h\pv{P}{s})$ is well-defined.  Conversely for
        every $x \in \conv{F}$ there exists a scheduler $\sch$ such that
        $t^{\sch,n}(h\pv{P}{s}) = x$.
\end{lemma}
\begin{proof}
      By an appeal to~\cite[Lemma 2.8]{tix09} one deduces that,
      \begin{align*}
                \begin{split}
                        \conv{F}
                & = \conv{ \left ( \bigcup_{
                        \nu \in \pv{P}{s} \longrightarrow}
                       \left \{ \left (\nu^S, h\pv{P}{s} \mapsto 1 \cdot \nu \right)
                        \right \} \times
                        \left (\textstyle{\sum_i}\, p_i \cdot \,
                                 F_i 
         + \Delta^n(\nu^S,\bot) \right ) \right ) }
                \\
                &
                =    \conv{ \left ( \bigcup_{
                        \nu \in \pv{P}{s} \longrightarrow}
                       \left \{ \left (\nu^S, h\pv{P}{s} \mapsto 1 \cdot \nu \right)
                        \right \} \times
                        \left (\textstyle{\sum_i}\, p_i \cdot \,
                                \conv{F_i}
         + \Delta^n(\nu^S,\bot) \right ) \right ) }
             \end{split}
        \end{align*}
        for all $n + 1$ with $n \in \Nats_+$.
  Both implications in the lemma's formulation  rely on this observation.  In
  particular the first implication then follows straightforwardly by induction
  over the positive natural numbers and by the trace functions' definition
  (Definition~\ref{defn:trace}).  As for the second implication, we build a
  scheduler $\sch$ by induction over the positive natural numbers in the
  following manner. For the base case suppose that we have $(\mu,c) \in \mathrm{conv}\, F
  \subseteq f_1(h\pv{P}{s})$. We can safely assume that it is of the form,
  \[
          (\mu,c) = \textstyle{\sum_k p_k}
          \cdot (\nu^S_k, h\pv{P}{s} \mapsto 1 \cdot \nu_k)
  \]   
  for some convex combination $\sum_k p_k \cdot (-)$. We thus set $\sch$ to be
  $\sum_k p_k \cdot (h\pv{P}{s} \mapsto 1 \cdot \nu_k)$ and it is clear that
  $t^{\sch,1}(h\pv{P}{s}) = (\mu,c)$ by the trace functions' definition.
  Regarding the inductive step we take a sequence $x \in \conv{F} \subseteq
  f_{n+1}(h\pv{P}{s})$. It is then the case that,
  \begin{align}
          \label{eq:nform}
          x = \textstyle{\sum_k p_k}
          \cdot \left ( (\nu^S_k, h\pv{P}{s} \mapsto 1 \cdot \nu_k)
                  \, \mathtt{::} \, 
                  \sum_{i} p_{k,i} \cdot x_{k,i}
                  + \Delta^n(\nu^S_k, \bot)
                  \right )
  \end{align}
  for some convex combination $\sum_k p_k \cdot (-)$ where $x_{k,i} \in
  \conv\, F_{k,i} \subseteq f_n(h\pv{P}{s}\nu_k\pv{P_{k,i}}{s_{k,i}})$ for all
  indices $k,i$. As in the proof of Proposition~\ref{main:theo}, we can assume that
  all valuations $\nu_k$ involved are pairwise distinct, by virtue of all sets
  $\conv{F_{k,i}}$ being convex and by recurring to the normalisation of
  subconvex combinations. Next, it follows from the induction hypothesis that for
  all traces $x_{k,i}$ there exists a scheduler $\sch_{k,i}$ such that,
  \[
          t^{\sch_{k,i},n}(h\pv{P}{s}\nu_k\pv{P_{k,i}}{s_{k,i}})
          = x_{k,i}
  \]
  So we set $\sch := \sum_k p_k \cdot (h\pv{P}{s} \mapsto 1 \cdot \nu_k) +
  \sum_{k,i} \sch_{k,i}$ and it is straightforward to show that the equation
  $t^{\sch,n+1}(h\pv{P}{s}) = x$ holds by an appeal to Lemma~\ref{lem:pref} and
  the fact that addition is monotone.
  \end{proof}
  The proof of the previous lemma provides a recipe for building a scheduler
  from any element of $\conv{F} \subseteq f_n(h\pv{P}{s})$.  The following
  lemma discloses useful properties about this construction.
\begin{lemma}
          \label{lem:mon}
          The scheduler construction described in the previous proof is
          functional, \ie\ exactly one scheduler arises from it. Furthermore
          the construction is monotone w.r.t. the prefix order of lists.
\end{lemma}

\begin{proof}
  Observe that the previous scheduler construction only relies on the choices
  $c_1\dots c_n$ in a given sequence $(\mu_1,c_1)\dots(\mu_n,c_n)$.
  Note as well that in order to not overburden notation, we use $c_1
  \dots c_n$ to denote a sequence obtained from a sequence in
  $\conv{F} \subseteq f_n(h\pv{P}{s})$ by an application of the map ${\pi_2}^n$.

  Let us start with the first claim. We just need to prove that the
  form~\eqref{eq:nform} that we assume from a sequence of choices is unique. We
  proceed by case distinction. Suppose that the sequence involved is $c_1$.
  We will show that there is only one expression of the form $\sum_k p_k
  \cdot (h\pv{P}{s} \mapsto 1 \cdot \nu_k)$ that can be equal to $c_1$. Thus
  consider another expression $\sum_j p'_j\cdot (h\pv{P}{s} \mapsto 1 \cdot
  \nu_j)$, \ie\ we change the convex combination involved (which is the only
  allowed change). Both convex combinations can be \emph{uniquely} determined
  by inquiring $c_1$ so indeed the expression is unique. Consider now 
  the case $c_1\dots c_{n+1}$. As before suppose that we have,
  \begin{align*}
          c_1\dots c_{n+1}  & = \textstyle{\sum_k p_k}
        \cdot \Big ( (h\pv{P}{s} \mapsto 1 \cdot \nu_k)
          \, \mathtt{::} \, 
          \sum_{i} p_{k,i} \cdot (d^1_{k,i} \dots 
          d^n_{k,i})
          \Big )
         \\
          c_1\dots c_{n+1}  & = \textstyle{\sum_j p'_j}
        \cdot \Big ( (h\pv{P}{s} \mapsto 1 \cdot \nu_j)
          \, \mathtt{::} \, 
          \sum_{i} p'_{j,i} \cdot (d'^1_{j,i} \dots 
          d'^n_{j,i})   
          \Big )
  \end{align*}
  Both convex combinations $\sum_k p_k \cdot (-)$ and $\sum_j p_j \cdot (-)$
  are equal for the same reason as before and thus the only elements
  in the expressions that could be different are $d^l_{k,i}$ and $d'^l_{k,i}$
  for some $k,i$ and some $1 \leq l \leq n$.  We will show that even these
  elements are always equal which proves our first claim.  First by the
  definition of $f_{n+1}(h\pv{P}{s})$ it is clear that for all $k,i$ the
  support of $d^l_{k,i}$ and $d'^l_{k,i}$ ($1 \leq l \leq n$) can only
  contain inputs with $h\pv{P}{s}\nu_k\pv{P_{k,i}}{s_{k,i}}$ as prefix.
  Then for every such input $a$ we have,
  \begin{align*}
          \textstyle{\sum_k}\, p_k \cdot \textstyle{\sum_{i}} 
          p_{k,i} \cdot d^l_{k,i}(a)
          =
          c_{l+1}(a)
          =
          \textstyle{\sum_i} \, p_k \cdot \textstyle{\sum_{i}} 
          p_{k,i} \cdot d'^l_{k,i}(a)
  \end{align*}
  which entails $p_k \, \cdot \, p_{k,i} \cdot d^l_{k,i}(a) = p_k \cdot
  p_{k,i} \cdot d'^l_{k,i}(a)$ for every $k,i$ and therefore $d^l_{k,i}(a) = d'^l_{k,i}(a)$.

  We now focus on the second claim, specifically we will prove that for all
  positive natural numbers $n \in \Nats_+$ and sequences, 
          \[
                x^1 \dots x^{n+1} \in
                \conv{F_{n+1}} \subseteq f_{n+1}(h\pv{P}{s})
          \]
  if $\sch$ and $\sch'$ are schedulers obtained respectively from $x^1 \dots
  x^n \in \conv{F_n} \subseteq f_n(h\pv{P}{s})$ and $x^1 \dots x^{n+1}$, via
  the previous scheduler construction, then the 
  inequation $\sch \leq \sch'$ holds.  The proof follows by induction over
  the positive natural numbers and by exploiting the fact that addition is
  monotone. The base case follows straightforwardly by unravelling the
  definition of $\conv{F_2}$.  For the inductive step it is clear that,
  \[
  c_1\dots c_{n+2} = \textstyle{\sum_k p_k}
  \cdot \left ( (h\pv{P}{s} \mapsto 1 \cdot \nu_k)
          \, \mathtt{::} \, 
          \sum_{i} p_{k,i} \cdot (d^1_{k,i},\dots,d^{n+1}_{k,i})
          \right )
  \]   
  for some convex combination $\sum_k p_k \cdot (-)$ and for sequences of
  choices $(d^1_{k,i},\dots,d^{n+1}_{k,i})$ obtained from sequences in $\conv{
  (F_{n+1})_{k,i}} \subseteq f_{n+1}(h\pv{P}{s}\nu_k\pv{P_{k,i}}{s_{k,i}})$
  by an application of the map ${\pi_2}^{n+1}$.  This entails the equation,
  \[
  c_1 \dots c_{n+1} = \textstyle{\sum_k p_k}
  \cdot \left ( (h\pv{P}{s} \mapsto 1 \cdot \nu_k)
          \, \mathtt{::} \, 
          \sum_{i} p_{k,i} \cdot (d^1_{k,i},\dots,d^{n}_{k,i})
          \right )
  \]   
  Now, in order to apply the induction hypothesis we need to show that for
  all indices $k,i$ the sequence $(d^1_{k,i},\dots,d^{n}_{k,i})$ can be
  obtained from a trace in  $\conv{ (F_{n})_{k,i}} \subseteq
  f_{n}(h\pv{P}{s}\nu_k\pv{P_{k,i}}{s_{k,i}})$ by an application of
  ${\pi_2}^n$. This follows directly from Lemma~\ref{lem:inc} and
  Lemma~\ref{lem:pref}. We thus derive $\sch_{k,i}' \leq \sch_{k,i}$ for all
  indices $k,i$ by the induction hypothesis, and finally we obtain the
  sequence of inequations,
  \begin{align*}
          \sch & = 
                       \textstyle{\sum_k p_k} \cdot (h\pv{P}{s} \mapsto 1
          \cdot \nu_k) + \sum_{k,i} \sch_{k,i}
           \\ & \leq
                \textstyle{\sum_k p_k} \cdot (h\pv{P}{s} \mapsto 1 \cdot \nu_k)
        + \sum_{k,i} \sch'_{k,i}
           \\ & = \sch'
  \end{align*}
\end{proof}

\begin{proposition}
        \label{prop:conv}
        Consider a program $P$, a state $s$, and an history $h$. Consider
        also the set $\conv{F} \subseteq f_n(h\pv{P}{s})$ for some number $n
        \in \Nats_+$. Then the order on $\conv{F}$ is discrete and therefore
        $\mathord{\downarrow} \conv{F} \cap \mathord{\uparrow} \conv{F} =
        \conv{F}$ by the anti-symmetry property of partial orders.
\end{proposition}

\begin{proof}

        Take two elements $(\mu_1,c_1)\dots(\mu_n,c_n)$ and
        $(\mu'_1,c'_1)\dots(\mu'_n,c'_n)$ in $\mathrm{conv}\, F$.  We will
        start by showing that the following implication holds.
        \[
               c_1 \dots c_n \geq c'_1 \dots c'_n \Longrightarrow
               c_1 \dots c_n = c'_1 \dots c'_n
        \]
        The latter follows by induction over the positive natural numbers $n
        \in \Nats_+$. For the base case both $c_1$ and $c'_1$ have mass $1$
        at $h\pv{P}{s}$ and mass $0$ everywhere else.  Thus the assumption
        $c_1 \geq c'_1$ entails $c_1 = c'_1$ by Theorem~\ref{theo:locin}. For
        the inductive step we use the definition of $f_{n+1}(h\pv{P}{s})$ to
        safely assume that,
        \begin{align*}
        c_1\dots c_{n+1}  = \textstyle{\sum_k p_k}
        \cdot \left (  (h\pv{P}{s} \mapsto 1 \cdot \nu_k)
          \, \mathtt{::} \, 
          \sum_{i} p_{k,i} \cdot ( d^1_{k,i} \dots 
           d^n_{k,i}) 
          \right )
         \\
         c'_1 \dots c'_{n+1}  = \textstyle{\sum_k p_k}
        \cdot \left ( (h\pv{P}{s} \mapsto 1 \cdot \nu_k)
          \, \mathtt{::} \, 
          \sum_{i} p_{k,i} \cdot (d'^1_{k,i} \dots 
          d'^n_{k,i})
          \right )
        \end{align*}
        such that for all indices $k,i$ the respective sequences of choices
        $d^1_{k,i} \dots d^n_{k,i}$ and $d'^1_{k,i} \dots d'^n_{k,i}$ are
        obtained from $\conv{F_{k,i}} \subseteq
        f_n(h\pv{P}{s}\nu_k\pv{P_{k,i}}{s_{k,i}})$ by an application of the
        projection map ${\pi_2}^n$.  Note that we can safely assume the same
        convex combination $\sum_k p_k \cdot (-)$ for both cases, because such
        combinations are uniquely encoded in both $c_1$ and $c'_1$ and we
        already know that $c_1 = c'_1$. Recall as well that we can safely
        assume that all the valuations $\nu_k$ involved in the combination are
        pairwise distinct, by virtue of the sets $\conv{F_{k,i}}$ being convex
        and by normalisation of subconvex valuations.

        The next step is to prove that $d^l_{k,i} \geq d'^l_{k,i}$ (for all $1
        \leq l \leq n$) to which we then apply the induction hypothesis to
        obtain $c_1\dots c_{n+1} = c'_1 \dots c'_{n+1}$. First by the
        definition of $f_{n}(h\pv{P}{s}\nu_k\pv{P_{k,i}}{s_{k,i}})$ it is clear
        that for all $k,i$ the support of $d^l_{k,i}$ and $d'^l_{k,i}$ ($1 \leq
        l \leq n$)  can only contain inputs with
        $h\pv{P}{s}\nu_k\pv{P_{k,i}}{s_{k,i}}$ as prefix.  Second for every
        such input $a$ we have,
        \begin{align*}
                \textstyle{\sum_k}\, p_k \cdot \textstyle{\sum_{i}}\, p_{k,i}
                \cdot d^l_{k,i}(a) = c_{l+1}(a) \geq c'_{l+1}(a) =
                \textstyle{\sum_k} \, p_k \cdot \textstyle{\sum_{i}}\,
                p_{k,i} \cdot d'^l_{k,i}(a) 
        \end{align*} 
        which entails $p_k \, \cdot \, p_{k,i} \cdot d^l_{k,i}(a) \geq p_k
        \cdot p_{k,i} \cdot d'^l_{k,i}(a)$ for every $k,i$ and thus
        $d^l_{k,i}(a) \geq d'^l_{k,i}(a)$.

        Consider again two elements $(\mu_1,c_1)\dots(\mu_n,c_n) \geq
        (\mu'_1,c'_1)\dots(\mu'_n,c'_n)$ of $\mathrm{conv}\, F$. From the
        previous reasoning we know that for all $1 \leq l \leq n$ we have $c_l
        = c'_l$, which entails that both lists are actually the same since each
        $\mu_l$ (resp. $\mu'_l$) only depends on the choices present in the
        respective list. This last claim follows from Lemma~\ref{lem:inc} and
        Lemma~\ref{lem:mon}.
\end{proof}
We finally reach the goal of this section.
\begin{proof}[Proof of Theorem~\ref{prop:alg}]
       Our proof  involves the use of codirected limits in the category
       $\CompHaus$ of compact Hausdorff spaces and continuous maps. 
       For every $n \in \Nats_+$ we set
       $X_n := f_n(\pv{P}{s})$, a compact Hausdorff space w.r.t. the Lawson
       topology of the $n$-fold product $(\PP(S) \times \Ch)^n$ in $\Coh$
       where $\PP(S) \times \Ch$ is also a $\Coh$-product. Given a DCPO
       $X$ let $L(X)$ denote the carrier of $X$
       equipped with respective Lawson topology. It then follows
       from~\cite[Proposition 5.1.56, Exercise 9.1.36, and Proposition
       9.3.1]{larrecq13} that the previous domain $(\PP(S) \times \Ch)^n$
       equipped with the respective Lawson topology is equal to $(L(\PP(S))
       \times L(\Ch))^n$  where in the latter case the constructs $(\times)$
       and $(-)^n$ are products in $\CompHaus$.  Let us now take the
       codirected limit ($\lim \Diag$) of,
       \[
               \xymatrix@C=35pt{
                       X_1 & \ar[l]_(0.45){\mathrm{lst}_1} X_2 &
                       \ar[l]_(0.55){\mathrm{lst}_2}  X_3\ \dots
               }
       \]
       in the category $\CompHaus$. Concretely each function $\mathrm{lst}_{n}:
       X_{n+1} \to X_n$ forgets the last pair of a given sequence (of pairs) in
       $X_{n+1}$. It is both well-defined and continuous because it is (up-to
       isomorphism) a restriction of a projection on the topological product
       $(L(\PP(S)) \times L(\Ch))^n$ and both Lemma~\ref{lem:pref} and
       Lemma~\ref{lem:inc} hold.  The space $\lim \Diag$ is thus compact
       Hausdorff by construction. Let us equip it with a concrete and more
       amenable characterisation. First by general topological results $\lim
       \Diag$ is concretely defined as the subspace,
       \[
               \lim \Diag \cong
               \left \{ (t_i)_{i \in \Nats_+} \in \textstyle{
                       \prod_{i \in \Nats_+}} X_i
               \mid \mathrm{lst}_{i}(t_{i+1}) = t_i \right \}
       \]
       Second there exists a map in $\CompHaus$ that for a given family
       $(t_i)_{i \in \Nats_+} \in \prod_{i \in \Nats_+} X_i $ it only keeps
       the last element of each $t_i$ $(i \in \Nats_+)$. Formally it is
       defined by,
       \[
        \langle f_i \rangle_{i \in \Nats_+} :  
        \textstyle{ \prod_{i \in
        \Nats_+} } X_i \to (L(\PP(S)) \times L(\Ch))^\omega
       \hspace{2.5cm}
               f_i = \begin{cases}
                       \pi_1 & \text{if } i = 1 \\
                       \pi_{i} \comp \pi_i & \text{otherwise}
               \end{cases}
       \]
       It is straightforward to prove by induction over the positive natural
       numbers that its restriction to $\lim \Diag$ is injective and
       therefore a monomorphim in $\CompHaus$. This restriction then yields
       an epimorphism $\lim \Diag \to \img \langle f_i \rangle_{i \in
       \Nats_+}$~\cite[Examples 14.23]{adamek09} in $\CompHaus$ which must
       also be a monomorphism by general categorical results.  Finally since
       $\CompHaus$ is balanced~\cite[Examples 7.50]{adamek09} we obtain an
       isomorphism $\lim \Diag \cong \img \langle f_i \rangle_{i \in
       \Nats_+}$. One can thus equivalently treat $\lim \Diag$ as a compact
       Hausdorff subspace of $(L(\PP(S)) \times L(\Ch))^\omega$ with each
       element of $\lim \Diag$ a sequence in $(L(\PP(S))
       \times L(\Ch))^\omega$ where every prefix of size $n$ is an element of
       $X_n$.  We now wish to show that,
       \begin{align}
               \lim \Diag \subseteq  \bigcup_{i \in \Nats_+} 
               (\PP(S) \times \Ch)^{i-1}
               \times (\varphi \times \Ch) \times (\PP(S) \times \Ch)^\omega 
               \label{eq:incl}
       \end{align}
       where $\varphi$ is treated as a Scott-open (note that the union is
       then an open cover).  We proceed by taking a sequence
       $(\mu_1,c_1)(\mu_2,c_2) \dots \in \lim \Diag$. By Lemma~\ref{lem:inc},
       Lemma~\ref{lem:mon}
       and Proposition~\ref{prop:conv} for every prefix
       $(\mu_1,c_1)\dots(\mu_n,c_n) \in X_n$ we obtain a scheduler $\sch_n$
       that yields the same trace. More generally we obtain a chain of schedulers
       $\sch_1 \leq \sch_2 \leq \dots$ in this way and denote the
       corresponding supremum by $\sch$.  Note that $\sch$ is non-blocking
       w.r.t. $\pv{P}{s}$ by virtue of Theorem~\ref{theo:tracef} and
       Lemma~\ref{lem:pref}.  Then by the assumption on non-blocking
       schedulers and again Theorem~\ref{theo:tracef} and
       Lemma~\ref{lem:pref} there must exist a natural number $k\in \Nats_+$
       such that $\mu_k \in \varphi$. This shows that the
       inclusion~\eqref{eq:incl} indeed holds. Now, since $\lim \Diag$ is
       compact we obtain,
       \[
               \lim \Diag \subseteq  
               \bigcup_{i \in F} (\PP(S) \times \Ch)^{i-1}
               \times (\varphi \times \Ch) \times (\PP(S) \times \Ch)^\omega  
       \]
       for $F \subseteq \Nats_+$ a finite set and we define $\mathbf{z} :=
       \max F$.  Finally it follows from Theorem~\ref{theo:tracef},
       Lemma~\ref{lem:pref}, and Lemma~\ref{lem:inc} that for every
       non-blocking scheduler $\sch$ the corresponding infinite sequence of
       traces must be in $\lim\, \Diag$ and therefore we obtain $\pv{P}{s}
       \Downarrow^{\sch,i} \mu_i$ with $\mu_i \in \varphi$ for some $i \in
       F$. It is then clear that $\pv{P}{s} \Downarrow^{\sch,\mathbf{z}}
       \mu_\mathbf{z}$ and $\mu_\mathbf{z} \in \varphi$ by
       Proposition~\ref{theo:det} and the fact that $\varphi$ is a Scott-open.
\end{proof}

\end{document}

%% file: macros.tex
\newcommand{\klcomp}{\star}
\newcommand{\parI}{\mathop{\bowtie}}
\newcommand{\seqI}{\mathop{\triangleright}}
\DeclareMathOperator{\demon}{\square}
\DeclareMathOperator{\angel}{\Diamond}
\makeatletter
\DeclareRobustCommand{\iscircle}{\mathord{\mathpalette\is@circle\relax}}
\newcommand\is@circle[2]{%
  \begingroup
  \sbox\z@{\raisebox{\depth}{$\m@th#1\bigcirc$}}%
  \sbox\tw@{$#1\square$}%
  \resizebox{!}{\ht\tw@}{\usebox{\z@}}%
  \endgroup
}
\makeatother
\DeclareMathOperator{\statt}{\iscircle_{\mathit{p}}}
\newcommand{\statto}[1]{\iscircle_{p_{#1}}}
\newcommand{\schfont}[1]{\mathcal{#1}}
\newcommand{\sch}{\schfont{S}}
\newcommand{\conv}[1]{\mathrm{conv}\, {#1}}

\newcommand{\catfont}[1]{\mathsf{#1}}

\newcommand{\catC}{\catfont{C}}

\newcommand{\CompHaus}{\catfont{CompHaus}}
\newcommand{\Subs}[2]{\catfont{Sub}_{}}
\newcommand{\Cone}{\catfont{Cone}}
\newcommand{\LCone}{\catfont{LCone}}

\newcommand{\Coh}{\catfont{CohDom}}

\newcommand{\Dcpo}{\catfont{DCPO}}
\newcommand{\Dom}{\catfont{Dom}}

\newcommand{\funfont}[1]{#1}
\newcommand{\funF}{\funfont{F}}

\newcommand{\sfunfont}[1]{\mathrm{#1}}
\newcommand{\Pow}{\sfunfont{P}}
\newcommand{\PP}{\sfunfont{V}}
\newcommand{\Dist}{\sfunfont{V}_{=1,\omega}}
\newcommand{\PDist}{\sfunfont{V}_{\leq 1,\omega}}

\newcommand{\Forg}[1]{\sfunfont{U}_{#1}}
\newcommand{\Id}{\sfunfont{Id}}

\newcommand{\Diag}{\mathscr{D}}

\newcommand\adjunct[2]{\xymatrix@=8ex{\ar@{}[r]|{\top}\ar@<1mm>@/^2mm/[r]^{{#2}}
& \ar@<1mm>@/^2mm/[l]^{{#1}}}}
\newcommand\adjunctop[2]{\xymatrix@=8ex{\ar@{}[r]|{\bot}\ar@<1mm>@/^2mm/[r]^{{#2}}
& \ar@<1mm>@/^2mm/[l]^{{#1}}}}
\newcommand\retract[2]{\xymatrix@=8ex{\ar@{}[r]|{}\ar@<1mm>@/^2mm/@{^{(}->}[r]^{{#2}}
& \ar@<1mm>@/^2mm/@{->>}[l]^{{#1}}}}
\newcommand{\pv}[2]{\langle #1, #2 \rangle}

\newcommand{\inl}{\mathrm{inl}}
\newcommand{\inr}{\mathrm{inr}}

\newcommand{\app}{\mathrm{app}}
\newcommand{\const}[1]{\underline{#1}}
\newcommand{\comp}{\cdot}
\newcommand{\id}{\mathrm{id}}
\newcommand{\unfold}{\mathrm{unfold}}
\newcommand{\fold}{\mathrm{fold}}

\newcommand{\sem}[1]{\left \llbracket #1 \right \rrbracket}
\newcommand{\asem}[1]{ \llparenthesis #1 \rrparenthesis}

\DeclareMathOperator{\img}{\mathrm{im}}

\newcommand{\Nats}{\mathbb{N}}
\newcommand{\Reals}{\mathbb{R}}
\newcommand{\Rats}{\mathbb{Q}}
\newcommand{\Rz}{\Reals_{\geq 0}}
\newcommand{\Complex}{\mathbb{C}}

\newcommand{\ie}{\emph{i.e.}}
\newcommand{\eg}{\emph{e.g.}}

\newcommand{\prog}[1]{{\mathtt{#1}}}

\newcommand{\upclos}{\mathord{\uparrow}}